\documentclass[journal]{IEEEtran}
\usepackage{mathrsfs}
\usepackage{graphicx, graphics, times, amsfonts, amsmath, amssymb, cite}
\usepackage{amsthm}
\usepackage{mathtools}
\usepackage{color}
\usepackage{hyperref}
\usepackage{multirow}
\usepackage{rotating}
\usepackage[font=small]{caption}
\usepackage{subcaption}
\usepackage{bm}
\let\forallsymb\forall
\input{mysymbol.sty}
\usepackage[utf8]{inputenc}
\usepackage{enumitem}
\usepackage[ruled,linesnumbered]{algorithm2e}
\usepackage[capitalise]{cleveref}

\usepackage{tikz, pgfplots}
\pgfplotsset{compat=1.17}
\pgfplotstableset{col sep=semicolon}
\usepackage{moresize} % provide more sizes for pgfplot
\usepgfplotslibrary{groupplots}
\usepgfplotslibrary{external}
\tikzset{every mark/.append style={scale=1.4, solid}, font=\footnotesize}
\pgfplotsset{
    width=1\textwidth,
    height=5.5cm,
    legend style={
        font=\ssmall ,  %\scriptsize,  %\ssmall,
        inner xsep=1pt,
        inner ysep=1pt,
        nodes={inner sep=1pt}},
    legend cell align=left,
    every axis/.append style={line width=0.5pt},
 	every axis plot/.append style={line width=1.25pt},
 	every axis y label/.append style={yshift=-4pt}
}

\newtheorem{corollary}{Corollary}
\newtheorem{theorem}{Theorem}
\newtheorem{proposition}{Proposition}
\newtheorem{problem}{Problem}
\newcommand{\vvec}{{\mathrm{vec}}}

\def\bchkH{{\ensuremath{\mathbf{\check H}}}}

\def\bchks{{\ensuremath{\mathbf{\check s}}}}

\renewcommand{\baselinestretch}{.93}

\begin{document}

\title{Robust Graph Filter Identification and Graph Denoising from Signal Observations}

\author{Samuel Rey,~\IEEEmembership{Student Member,~IEEE},
        Victor M. Tenorio, %~\IEEEmembership{Student Member,~IEEE},  %Victor, te quito la membresía porque así ahorramos una línea, en la versión final la ponemos
        and~Antonio G. Marques,~\IEEEmembership{Senior Member,~IEEE} \vspace{-0.4cm} % <-this % stops a space
\thanks{S. Rey, V. M. Tenorio and A. G. Marques are with the Department
of Signal Theory and Comms., King Juan Carlos University, Madrid, Spain, \{samuel.rey.escudero,victor.tenorio,antonio.garcia.marques\}@urjc.es.}% <-this % stops a space
\thanks{A conference version of this manuscript with some preliminary results was published in~\cite{rey2021robust}. Work in this paper was partially supported by the Spanish Grants SPGRAPH (PID2019-105032GB-I00), FPU17/04520 and FPU20/05554. The authors would like to thank Pablo Espana-Gutierrez for his help reviewing the literature and developing the early versions of the convergence results.}}


\maketitle

\vspace{-1.0cm}
\begin{abstract}
When facing graph signal processing tasks, the workhorse assumption is that the graph describing the support of the signals is known.
However, in many relevant applications \emph{the available graph suffers from observation errors and perturbations}.
As a result, any method relying on the graph topology may yield suboptimal results if those imperfections are ignored.  
Motivated by this, we propose a novel approach for handling perturbations on the links of the graph and apply it to the problem of robust graph filter (GF) identification from input-output observations.
Different from existing works, we formulate a non-convex optimization problem that operates in the vertex domain and jointly performs GF identification and graph denoising.
%As a result, on top of estimating the desired GF at hand, a modified (true) graph is obtained as a byproduct.
As a result, on top of learning the desired GF, an estimate of the graph is obtained as a byproduct.
To handle the resulting bi-convex problem, we design an algorithm that blends techniques from alternating optimization and majorization minimization, showing its convergence to a stationary point.
The second part of the paper i) generalizes the design to a robust setup where several GFs %(all defined over the same graph) 
are jointly estimated, and ii) introduces an alternative algorithmic implementation that reduces the computational complexity. %when dealing with large graphs.
Finally, the detrimental influence of the perturbations and the benefits resulting from the robust approach are numerically analyzed over synthetic and real-world datasets, comparing them with other state-of-the-art alternatives. 
\end{abstract}


\section{Introduction}
% Motivate GFs
\IEEEPARstart{N}{owadays}, a significant number of datasets are defined over an irregular (heterogeneous) support that can be conveniently represented by a graph. As a result, the data at hand can be readily understood as graph signals (alternatively, network processes) whose structure and properties depend on the topology of the generating graph.
Illustrative examples include measurements from power, communications, social, biological, or financial networks, to name a few~\cite{nodop1998field,kolaczyk2009book,sporns2012book,shuman2013emerging,sandryhaila2014discrete}.
Characterizing and modeling graph and network processes entails a prevalent and relevant task that not only enhances our understanding of the data at hand, but also opens the door to more sophisticated processing and knowledge extraction schemes.
A popular approach within the graph signal processing (GSP) framework~\cite{shuman2013emerging,djuric2018cooperative,iglesias2018demixing,ortega2018graph,marques2020editorial} is to represent the generative process by a \emph{graph filter} (GF) and, then, model the data as graph signals generated by applying a (low-pass, smooth, bandlimited...) GF to a simple (sparse, white, constant...) input. GFs are topology-aware linear operators whose output can be interpreted as the outcome of a network diffusion or spreading process~\cite{sandryhaila2013discrete,segarra2017optimal,liu2018filter}.
GFs can be expressed as polynomials of the graph-shift operator (GSO), a matrix encoding the topology of the graph, and on top of its theoretical interest, the task of GF identification is practically relevant to, e.g., understanding the dynamics of network diffusion processes \cite{segarra2017optimal,segarra2017blind,djuric2018cooperative}, as well as explaining the structure of real-world datasets~\cite{rey2019sampling,zhu2020estimating,he2022detecting}.

% Motivate accounting for errors
Since GSP is a relatively recent area of research, it is not surprising that most GSP works focus on how to harness the graph structure while assuming that the topology of the graph is perfectly known.
Nevertheless, this assumption is unlikely to hold in many practical setups, where graphs suffer from \emph{imperfections and perturbations}.
When networks are given explicitly, perturbations may be due to observational noise and errors (e.g., link failures in power or wireless networks~\cite{isufi2017filtering}).
When in lieu of physical entities, the graphs model (statistical) pairwise relationships among the observed variables, they need to be inferred from the data~\cite{friedman2008sparse,dong2016learning,segarra2017network,buciulea2022learning}.
Since this is a challenging (oftentimes ill-posed) task, the estimated graphs inherit the imperfections (limitations) of the graph learning scheme adopted (e.g., the thresholding operation implemented in correlation networks~\cite{kolaczyk2009book}).
Intuitively, the presence of perturbations hinders any GSP scheme or GSP tool applied to the data. While this is true regardless of the task and hand, it is even more relevant for those involving spectral transforms and GFs, since eigenvectors and high-order matrix polynomials are more sensitive to errors in the matrix codifying the graph.

% Paper contribution
To address these challenges, this paper 
%builds upon the preliminary work in~\cite{rey2021robust} and 
investigates the problem of estimating a GF from input-output signal pairs assuming that both the signals and the supporting graph have errors.
The proposed approach is formulated in the vertex domain, avoiding the numerical instability of computing large polynomials and, at the same time, bypassing the challenges associated with robust \emph{spectral} graph theory. 
We recast the robust estimation as a joint optimization problem where the GF identification objective is augmented with a graph-denoising regularizer, so that, on top of the desired GF, we also obtain an enhanced estimate of the supporting graph.
The joint formulation leads to a non-convex bi-convex optimization problem, for which a provably-convergent efficient (alternating minimization) algorithm able to find an approximate solution is developed. Furthermore, to address scenarios where multiple GFs are present (e.g., when dealing with vector autoregressive (AR) spatio-temporal processes or in setups where nodes collect multi-feature vector measurements), we generalize our framework so that multiple GFs, all defined over the same graph, are jointly identified. 

% Challenging task - brief review
Despite their theoretical and practical relevance, the number of robust GSP works is limited, due in part to the challenges emanating from the presence of graph perturbations~\cite{ceci2020graph,segarra2015stability,miettinen2019modelling,ceci2020_semtls}.
In the spectral domain,\cite{ceci2020graph} employs a small perturbation analysis to study the impact of perturbations in the spectrum of the graph Laplacian.
In the vertex domain,~\cite{miettinen2018graph,miettinen2019modelling} postulates a graphon-based perturbation model and analyzes how perturbations affect GFs of order one.
More recently, \cite{ceci2020_semtls} combines structural equation models (SEMs) with total least squares (TLS) to jointly infer the GF and the perturbations when the observed data is explained by a SEM.
A different robust alternative is presented in~\cite{natali2020topology}, where the support of the graph is assumed to be known and the goal is to estimate the weights of the network topology and the coefficients of the GF.
The resultant problem is non-convex and the authors adopt a sequential convex programming (SCP) approach to solve it.
Finally, the presence of perturbations has also been considered in non-linear GSP tasks.
An alternative definition of GFs robust to perturbations is proposed in~\cite{tenorio2021robust}, and the transferability of GFs when employed in graph neural networks is studied in~\cite{levie2019transferability,levie2021transferability,ruiz2021graph}.

% Contributions and outline
\vspace{2mm}
\noindent\textbf{Contributions and outline.}
We close the section by summarizing the organization and contributions of the manuscript.
After reviewing preliminary GSP concepts in Sec.~\ref{S:preliminaries}, we analyze the influence of edge perturbations in polynomial GFs and state the robust GF identification problem in Sec.~\ref{S:perturbed_graph_filters}.
After that, our main contributions are:
\begin{enumerate}%[leftmargin=3mm]
    \item We formulate a non-convex optimization problem to jointly estimate the graph and the GF, develop an alternating optimization algorithm to solve it and prove its convergence to a stationary point (Sec.~\ref{S:rfi}).
    \item We consider a generalization where several GFs are jointly estimated by exploiting the fact that they are polynomials of the same GSO (Sec.~\ref{S:rfi_joint}).
    \item We propose an efficient implementation of the GF identification algorithm to handle graphs with a large number of nodes (Sec.~\ref{S:efficient_rfi}). 
\end{enumerate}
The effectiveness of the proposed algorithms is evaluated numerically in Sec.~\ref{S:experiments}, and some concluding remarks are provided in Sec.~\ref{S:conclusion}. Last but not least, while we focus on GF identification from input-output pairs, the approach put forth in this paper can be generalized to other GSP tasks, which is a research path we plan to pursue in the near future.

%%%%%%%%%%%%%%%%%%%%%%%%%%%%%%%%%%%%%%%%%%%%%%%%%%%%%%%%%%%%%%%%%%%%%%%%%%%%%%%%%%%%%%%%%%%%%%%%%%%%%%%%%%%%%%%%%
%SECTION
%%%%%%%%%%%%%%%%%%%%%%%%%%%%%%%%%%%%%%%%%%%%%%%%%%%%%%%%%%%%%%%%%%%%%%%%%%%%%%%%%%%%%%%%%%%%%%%%%%%%%%%%%%%%%%%%%
% Alternative title: GSP for filter identification
\section{GSP preliminaries}\label{S:preliminaries}
\vspace{1mm}
\noindent\textbf{Graphs and graph signals.}
Consider a directed graph $\ccalG=(\ccalV, \ccalE)$ formed by the set of $N$ nodes (vertices) collected in $\ccalV$ and the set of edges $\ccalE \subset \ccalV\times \ccalV$, such that $(i,j) \in \ccalE$ if node $i$ is connected to node $j$.
The topology of $\ccalG$ can be represented by the adjacency matrix $\bbA \in \reals^{N \times N}$, a sparse matrix whose elements $A_{ij}$ are non-zero if and only if $(i,j) \in \ccalE$.
When $\ccalG$ is weighted, the entries $A_{ij}$ capture the strength of the link between nodes $i$ and $j$.
Otherwise, the elements of $\bbA$ are binary.
Together with the graph $\ccalG$, we focus on modeling the data as signals defined on the nodes in $\ccalV$.
Formally, a (nodal) \emph{graph signal} is a function from the vertex set to the real field $x:\ccalV \to \reals$, which can be alternatively represented as an $N$-dimensional vector $\bbx \in \reals^N$ whose $i$-th entry $x_i$ denotes the value of the signal at node $i$.
The foundational assumption of GSP is that the properties of the $\bbx$ and the topology of $\ccalG$ are related.
As a simple example, if $\ccalG$ captures the similarity between nodes and $A_{ij}$ is high, then the values of $x_i$ and $x_j$ are expected to be akin to each other. 

\vspace{2mm}
\noindent\textbf{Graph-shift operator (GSO).}
The GSO associated with a graph $\ccalG$ of $N$ nodes is as a generic matrix $\bbS \in \reals^{N \times N}$ that: i) captures the topology of the underlying graph and ii) represents a \emph{local} and \emph{linear} transformation that can be applied to graph signals defined over $\ccalG$. The entries of $\bbS$ satisfy $S_{ij} \neq 0$ only if $(i,j) \in \ccalE$ or $i=j$ and, as a result, the application of the GSO to a graph signal involves mixing values among one-hop neighborhoods. Two typical choices for the GSO are the adjacency matrix $\bbA$ and the graph Laplacian $\bbL :=\! \diag(\bbA\textbf{1}) \!-\! \bbA$~\cite{shuman2013emerging,djuric2018cooperative}, where $\diag(\cdot)$ transforms a vector into a diagonal matrix and $\textbf{1}$ is the vector of all ones.
We assume that $\bbS$ is diagonalized as $\bbS = \bbV\diag(\bblambda)\bbV^{-1}$, where $\bbV$ is an $N \times N$ matrix collecting the eigenvectors of $\bbS$, and vector $\bblambda$ collects its eigenvalues.
The matrix $\bbV^{-1}$ is commonly adopted as the Graph Fourier Transform (GFT) for graph signals with $\tbx = \bbV^{-1}\bbx$ denoting the graph frequency representation of $\bbx$~\cite{sandryhaila2014discrete}.

\vspace{2mm}
\noindent\textbf{Graph filtering.}
GFs are topology-aware operators whose inputs and outputs are graph signals. More specifically, GFs implement a linear transformation that can be written as a polynomial of $\bbS$
% as a matrix polynomial of the GSO $\bbS$ 
\begin{equation}\label{eq:graph_filter}
     \bbH = \sum_{r=0}^{N-1} h_r\bbS^r =\bbV\diag(\bbPsi\bbh)\bbV^{-1} = \bbV\diag(\tbh)\bbV^{-1},
\end{equation}
where $\bbh = [h_0,...,h_{N-1}]$ is the vector collecting the GF coefficients. The $N \times N$ Vandermonde matrix $\bbPsi$ defined as $\Psi_{ij} := \lambda_{i}^{j-1}$ represents the GFT for GFs, and thus, $\tbh = \bbPsi\bbh$ is the vector of size $N$ representing the frequency response of $\bbH$~\cite{sandryhaila2014discrete,segarra2017optimal}.
Since $\bbS^r$ encodes the $r$-hop neighborhood of the graph, a graph signal given by $\bby = \sum_{r=0}^{N-1} h_r\bbS^r\bbx = \bbH\bbx$ can be interpreted as a version of the input signal $\bbx$ diffused across $N\!-\!1$ neighborhoods with $h_r$ being the coefficients of the linear combination~\cite{segarra2017optimal}. Moreover, there are scenarios where $h_r=0$ for $r\geq R$. In those cases, the order of the filter is $R$ and, if more convenient, $\bbh$ and $\bbPsi$ can be redefined so that the number of elements of $\bbh$ (columns of $\bbPsi$) is $R$ in lieu of $N$.

\vspace{2mm}
\noindent\textbf{Graph stationarity.}
A zero-mean \emph{random} graph signal $\bbx$ is said to be stationary on $\ccalG$ if its covariance matrix $\bbC_\bbx = \mathbb{E}[\bbx\bbx^\top]$ is a positive-semidefinite polynomial of the GSO \cite{marques2017stationary}.
%Consider now an undirected graph $\ccalG$, so that $\bbS$ is symmetric and $\bbV\bbV^\top = \bbI$.
%Then, assuming that $\bbx$ is stationary is tantamount to assuming that its covariance $\bbC_\bbx$ and the matrix $\bbS$ share the same orthogonal eigenvectors $\bbV$ \cite{marques2017stationary}.
A common example of stationary graph signals arises when $\bbx$ is the output of a linear graph diffusion process whose input (initial condition) is a white signal $\bbw$ and whose diffusion dynamics can be accurately represented by a GF. 
Mathematically, if we have that $\bbx = \bbH\bbw$ with $\bbH$ being a GF [cf. \eqref{eq:graph_filter}] and $\bbC_\bbw = \mathbb{E}[\bbw\bbw^\top] = \bbI$, it follows that the covariance of $\bbx$ is $\bbC_\bbx = \bbH\bbH^\top=\sum_{r=0,r'=0}^{N-1,N-1}h_r h_{r'} \bbS^{r+r'}$. Since the latter is a polynomial of the GSO, it follows that $\bbx$ is stationary on $\ccalG$.

\subsection{GF identification from input-output pairs}\label{S:graph_filter_id}
% Description of GFI problem
In the context of linear operators, let us consider that we observe $M$ input and output pairs $\bbX := [\bbx_1,...,\bbx_M]$ and $\bbY := [\bby_1,...,\bby_M]$ whose relation is given by 
%a diffusion process of the form of $\bbY = \bbH\bbX + \bbW$,
\begin{equation}\label{eq:observation_model}
    \bbY = \bbH\bbX + \bbW,
\end{equation}
with $\bbW$ being a zero-mean random matrix (typically assumed to have i.i.d. entries) that accounts for noisy measurements and model inaccuracies. Leveraging \eqref{eq:observation_model}, the GF identification task amounts to using the input-output pairs to estimate $\bbH$ under the model in \eqref{eq:graph_filter}, which, if the GSO $\bbS$ is known, boils down to estimating the GF coefficients collected in $\bbh \in \reals^N$.

% Solution of the GFI
Hence, we can approach the GF identification task in the node domain by solving the convex problem
\begin{equation}\label{eq:fi_opt}
    \min_\bbh \big\|\bbY - \sum_{r=0}^{N-1}h_r\bbS^r\bbX \big\|_F^2.
\end{equation}
% We can approach the GF identification task in the node domain by minimizing the convex loss function $\|\bbY - \sum_{k=0}^{K-1}h_k\bbS^k\bbX \|_F^2$.
Leveraging the frequency definition of GFs in \eqref{eq:graph_filter}, we use the GFT matrices $\bbV^{-1}$ and $\bbPsi$ to rewrite the least-squares (LS) cost in \eqref{eq:fi_opt} and obtain its (closed-form) solution as
\begin{equation}\label{eq:solving_fi}
    \!\!\!\hbh \!=\! \argmin_\bbh \|\vvec(\bbY) \!-\! (\!(\bbV^{-1}\bbX)^\top \!\!\odot\! \bbV)\bbPsi\bbh \|_2^2 \!=\! \bbTheta^{\dagger}\vvec(\bbY),
\end{equation}
where $\vvec(\cdot)$ denotes the vectorization operation, $\bbV^{-1}\bbX$ is the frequency representation of the input signals, $\odot$ denotes the Khatri–Rao product, $\bbPsi$ is the GFT Vandermonde matrix, $\bbTheta:=((\bbV^{-1}\bbX)^\top \odot \bbV)\bbPsi$, and $^{\dagger}$ is the pseudoinverse operator.

% Comments on the solution
From \eqref{eq:solving_fi} we observe that estimating $\bbH$ is straightforward under the assumptions of: i) $\bbTheta$ being full rank (i.e., the inputs are sufficiently rich) and ii) $\bbS$ being perfectly known.
However, the (critical) assumption in ii) does not hold in most practical settings. The remainder of the paper approaches the GF identification problem assuming imperfect knowledge of the GSO.

%%%%%%%%%%%%%%%%%%%%%%%%%%%%%%%%%%%%%%%%%%%%%%%%%%%%%%%%%%%%%%%%%%%%%%%%%%%%%%%%%%%%%%%%%%%%%%%%%%%%%%%%%%%%%%%%%
%SECTION
%%%%%%%%%%%%%%%%%%%%%%%%%%%%%%%%%%%%%%%%%%%%%%%%%%%%%%%%%%%%%%%%%%%%%%%%%%%%%%%%%%%%%%%%%%%%%%%%%%%%%%%%%%%%%%%%%
%\section{Modeling graph perturbations}
\section{GF identification with imperfect graph knowledge}\label{S:perturbed_graph_filters}
% Perturbed GSO
This section introduces and discusses the problem of estimating a GF $\bbH=\sum_{r=0}^{N-1}h_r\bbS^r$ from noisy input-output signal pairs $(\bbX\in \reals^{N \times M}, \bbY \in \reals^{N \times M})$ assuming that we have access to an \emph{imperfect} GSO $\barbS \in \reals^{N \times N}$, which can be modeled as 
\begin{equation}\label{eq:additive_GSO_perturbation_model}
\barbS = \bbS + \bbDelta,
\end{equation}
where $\bbS \in \reals^{N \times N}$ represents the true GSO and $\bbDelta \in \reals^{N \times N}$ is a \emph{perturbation matrix}. 
Before discussing models for the perturbation matrix, we find illustrative to demonstrate the impact of $\bbDelta$ on the GSP problem at hand.

%%%%%%%%%%   FIGURE   %%%%%%%%%%
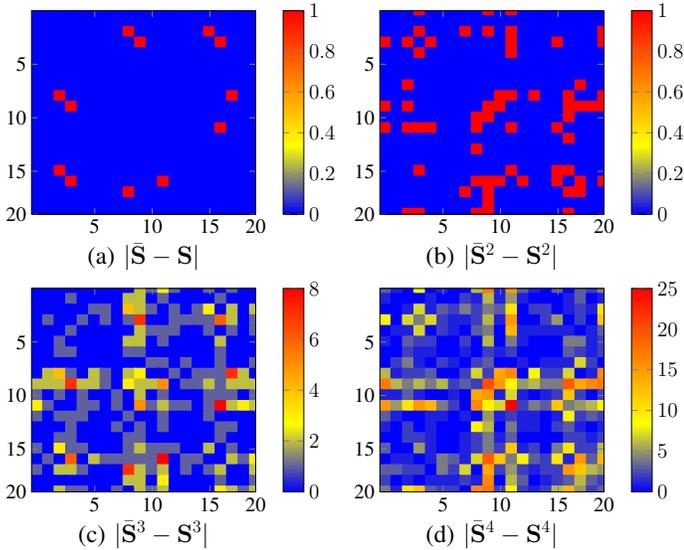
\begin{figure}[!t]
	\centering
	\begin{tikzpicture}[baseline,scale=.55]
\begin{groupplot}[
    table/col sep=space,
    width=7cm,
    height=6.5cm,
    group style={group size=2 by 2,
        horizontal sep=3cm,
        vertical sep=1.8cm,},
    enlargelimits=false,
    colorbar,
    % point meta max=15,
    % colorbar style={ytick={0,5,...,15}},
    xtick={5,10,15,19},
    xticklabels={5,10,15,20},
    ytick={5,10,15,19},
    yticklabels={5,10,15,20},
    x label style={font=\LARGE},
    tick label style={font=\Large}
    ]

    \pgfplotstableread{data/A-Apert_k1.csv}\matrixA
    \pgfplotstableread{data/A-Apert_k2.csv}\matrixB
    \pgfplotstableread{data/A-Apert_k3.csv}\matrixC
    \pgfplotstableread{data/A-Apert_k4.csv}\matrixD
    
% Known bug that affects visualization when reading from externa file and using matrix plot
\nextgroupplot[xlabel={(a) $|\barbS - \bbS|$},]
    \addplot [
        matrix plot,
        % nodes near coords=\coordindex, mark=*,
        point meta=explicit,
        mesh/cols=20,  % only if the matrix is not square
    ] table [meta=value] {\matrixA};

\nextgroupplot[xlabel={(b) $|\barbS^2 - \bbS^2|$}]
    \addplot [
        matrix plot,
        point meta=explicit,
        mesh/cols=20,  % only if the matrix is not square
    ] table [meta=value] {\matrixB};

\nextgroupplot[xlabel={(c) $|\barbS^3 - \bbS^3|$}]
    \addplot [
        matrix plot,
        point meta=explicit,
        mesh/cols=20,  % only if the matrix is not square
    ] table [meta=value] {\matrixC};
    
\nextgroupplot[xlabel={(d) $|\barbS^4 - \bbS^4|$},
    colorbar style={ytick={0,5,...,25}}]
    \addplot [
        matrix plot,
        point meta=explicit,
        mesh/cols=20,  % only if the matrix is not square
    ] table [meta=value] {\matrixD};
    
\end{groupplot}
\end{tikzpicture}
	\caption{Absolute error for different powers of the matrix $\bbS$ and its perturbed version $\barbS$. The true GSO is the adjacency matrix of an Erd\H{o}s-Rényi graph with link probability of 0.15, and $\barbS$ is perturbed by creating and destroying links independently with a probability of $0.05$.}\label{fig:pert_err_example}
\end{figure}
%%%%%%%%%%   END FIGURE   %%%%%%%%%%

% Impact of perturbations in GF
As pointed out in the introduction, the presence of uncertainties in the topology of $\ccalG$ is particularly relevant when dealing with GFs.
Indeed, due to the polynomial definition of $\bbH$, %[see \eqref{eq:graph_filter}]
even small perturbations can lead to significant errors when $\barbS$ (and not $\bbS$) is used as the true GSO.
To see this more clearly, \cref{fig:pert_err_example} provides an example that illustrates how the errors encoded in $\bbDelta$ propagate for different matrix powers, demonstrating that the discrepancies between $\barbS^r$ and $\bbS^r$ increase swiftly as the power $r$ grows.
More rigorously, let $C$ be a positive constant such that $\|\bbS\|\leq C$ and $\|\barbS\| \leq C$, and define %$\bbH:=\sum_{r=0}^{N-1}h_r\bbS^{r}$ and
$\barbH:=\sum_{r=0}^{N-1}h_r\barbS^{r}$. 
%Then, it follows that the error generated by the perturbations is upper-bounded by
Then, the error generated by the perturbations is upper-bounded by
\begin{equation}\label{eq:err_H_pert}
    \|\barbH - \bbH\| \leq \sum_{r=1}^{N-1}|h_r|\|\barbS^r - \bbS^r\| \leq \sum_{r=1}^{N-1}|h_r|rC^{r-1}\|\bbDelta\|,
\end{equation}
where the last inequality follows from~\cite[Lemma 3]{levie2019transferability}.
In words, the maximum difference between the true $\bbH$ and the perturbed $\barbH$ increases \emph{exponentially} with the degree of the GF.

% Need for a robust approach
From the previous discussion, it is not surprising that the imperfect knowledge of the graph topology is also relevant when estimating the filter coefficients.
In fact, ignoring the errors in $\bbDelta$ and attempting to estimate $\bbh$ solving \eqref{eq:solving_fi} when $\barbS$ is used in lieu of the true (unknown) $\bbS$ leads to a poor solution, as we illustrate numerically in Sec.~\ref{S:experiments}.
Motivated by this, we approach the GF identification problem from a robust perspective by taking into account the imperfect knowledge of the GSO.
The resultant robust estimation task is formally stated next.

% \begin{problem}\label{P:problem_statement}
%     Let $\ccalG$ be a graph with $N$ nodes, matrix $\bbS \in \reals^{N \times N}$ be the true (unknown) GSO associated with $\ccalG$, and matrix $\barbS \in \reals^{N \times N}$ be the perturbed (observed) GSO. Moreover, let $\bbX\in \reals^{N \times M}$ and $\bbY \in \reals^{N \times M}$ be a pair of matrices collecting $M$ observed input and output signals defined over the graph $\ccalG$ and related by the model in \eqref{eq:observation_model}.
%     %Let $\bbS \in \reals^{N \times N}$ denote the true (unknown) GSO and $\barbS = \bbS + \bbDelta$ the perturbed (observed) GSO, with $\bbDelta$ capturing the discrepancies between $\bbS$ and $\barbS$.
%     Our goal is to use the triplet $(\bbX,\bbY,\barbS)$ to: i) learn the GF $\bbH$ that best fits the model in \eqref{eq:observation_model} and ii) recover an enhanced estimation of the true GSO $\bbS$.
%     To that end, we make the following assumptions: \\
%     (\textbf{AS1}) $\bbH$ is a polynomial of $\bbS$ [cf. \eqref{eq:graph_filter}]. \\
%     (\textbf{AS2}) $\bbS$ and $\barbS$ are close according to some metric $d(\bbS,\barbS)$, i.e., the observed perturbations are ``small'' in some sense.
% \end{problem}

\begin{problem}\label{P:problem_statement}
    Let $\ccalG$ be a graph with $N$ nodes, let $\bbS \in \reals^{N \times N}$ be the true (unknown) GSO associated with $\ccalG$, and let $\barbS \in \reals^{N \times N}$ be the perturbed (observed) GSO. Moreover, let $\bbX\in \reals^{N \times M}$ and $\bbY \in \reals^{N \times M}$ be a pair of matrices collecting $M$ observed input and output signals defined over $\ccalG$ and related by the model in \eqref{eq:observation_model}.
    Our goal is to use the triplet $(\bbX,\bbY,\barbS)$ to: i) learn the GF $\bbH$ that best fits the model in \eqref{eq:observation_model} and ii) recover an enhanced estimation of $\bbS$.
    To that end, we make the following assumptions: \\
    (\textbf{AS1}) $\bbH$ is a polynomial of $\bbS$ [cf. \eqref{eq:graph_filter}]. \\
    (\textbf{AS2}) $\bbS$ and $\barbS$ are close according to some metric $d(\bbS,\barbS)$, i.e., the observed perturbations are ``small'' in some sense.
\end{problem}

% Comment problem and assumptions
On top of the previous two assumptions, we also consider that the norm of the noise observation matrix $\bbW$ in \eqref{eq:observation_model} is small, which is a workhorse assumption in this type of problems. Similar to standard GF identification approaches, (\textbf{AS1}) limits the degrees of freedom of the linear operator in \eqref{eq:observation_model}. However, the fact of the true $\bbS$ being unknown adds  uncertainty to the problem and, as a result, additional signal observations are required to achieve an identification performance comparable to the one obtained when $\barbS=\bbS$. Regarding the recovery of the true GSO, (\textbf{AS2}) accounts for the hypothesis that $\barbS$ is a perturbed observation of $\bbS$ and, hence, matrices $\bbS$ and $\barbS$ are not extremely different. Note that this guarantees that ``some'' information about the true GSO is available, so that (\textbf{AS1}) can be effectively leveraged. While not exploited in our formulation, additional assumptions constraining the GSO could also be incorporated into the problem. Finally, the metric $d(\cdot, \cdot)$ employed to quantify the similarity between $\bbS$ and $\barbS$ should depend on the model for the perturbation $\bbDelta$, a subject that is briefly discussed next.  %, as is detailed next.

\subsection{Modeling graph perturbations}\label{S:graph_pert}
The development and analysis of graph perturbation models that combine practical relevance and analytical tractability constitutes an interesting yet challenging open line of research~\cite{miettinen2018graph,miettinen2019modelling}.
Due to its flexibility and tractability, this paper considers an additive perturbation model [cf. \eqref{eq:additive_GSO_perturbation_model}], so that the focus is constrained to understanding the structural (statistical) properties of matrix  $\bbDelta = \barbS - \bbS$.

% Creation/Destruction of links
Consider first the case where perturbations only \emph{create or destroy links} independently.
If $\ccalG$ is an \emph{unweighted graph}, a simple approach is to consider perturbations modeled as independent Bernoulli variables with possibly different creation/destruction probabilities.
In this case, the entries of $\bbDelta$ would be
\begin{equation}
    \Delta_{ij} = \left\{\hspace{-2mm} \begin{array}{rl}
        1 &  \mathrm{if} \; \mathrm{link} \; (i,j) \; \mathrm{is} \; \mathrm{created},\\
        -1 & \mathrm{if} \; \mathrm{link} \; (i,j) \; \mathrm{is} \; \mathrm{destroyed},\\
        0 & \mathrm{otherwise}.
    \end{array} \right.
\end{equation}
Since $\bbDelta$ models the creation and destruction of links, it is worth noting that $\Delta_{ij} = 1$ only if $S_{ij} = 0$ and $\Delta_{ij} = -1$ only if $S_{ij} = 1$.
In the more general case of $\ccalG$ being a \emph{weighted graph}, $\Delta_{ij} = -S_{ij}$ destroys an existing link while $\Delta_{ij} = z$ creates a new link.
Here, $z$ is a random variable sampled from a particular distribution (typically mimicking the weight distribution of the true $\bbS$).
When facing this type of perturbations, a suitable distance function is the $\ell_0$ norm
\begin{equation}
    d(\bbS,\barbS) = \| \bbS - \barbS \|_0,
\end{equation}
with the $\ell_1$ norm $\| \bbS - \barbS \|_1$ being a prudent convex relaxation.

Alternatively, rather than creating or destroying links, perturbations may represent uncertainty over the edge weights.
This entails the support of matrix $\bbDelta$ matching that of $\bbS$ and $\barbS$, and the non-zero entries of $\bbDelta$ being sampled from a distribution that models the observation noise. 
For example, if the noise is zero-mean, Gaussian and white, it holds that $\Delta_{ij} \sim \ccalN(0,\sigma^2)$ when $\bbS_{ij} \neq 0$ and $\Delta_{ij}=0$ when $\bbS_{ij} = 0$.
Under this setting, an appropriate distance metric is given by
\begin{equation}
    d(\bbS, \barbS) = \| \bbS_\ccalE - \barbS_\ccalE\|_2^2,
\end{equation}
where $\bbS_\ccalE$ and $\barbS_\ccalE$ only select the non-zero entries (edges) in $\bbS$ and $\barbS$.
Additionally, one can have setups where the two types of perturbations are simultaneously present.
That is, perturbations may create and destroy links while the actual value of the existing links is also uncertain.
In such a case, a combination of $\ell_1$  and $\ell_2$ norms like in elastic nets~\cite{zou2005regularization} is adequate.  

% Comment on non independent errors
The models previously described only consider the perturbation of edges in an independent fashion.
However, there may be scenarios where the perturbations are correlated.
Consider for example a communication network.
If the power supply of a node stalls, the signal-to-noise ratio of all its links will be poor, and hence, links involving that node will be more likely to fail.
Perturbations dependent across links can be modeled by means of a multivariate correlated Bernoulli distribution, an Ising model, or more sophisticated random graph models~\cite{dai2013multivariate}.
When prior information about the dependence of the perturbations is available, it can be incorporated into the function $d(\bbS,\barbS)$ to better extract the information encoded in $\barbS$.

%%%%%%%%%%%%%%%%%%%%%%%%%%%%%%%%%%%%%%%%%%%%%%%%%%%%%%%%%%%%%%%%%%%%%%%%%%%%%%%%%%%%%%%%%%%%%%%%%%%%%%%%%%%%%%%%%
%SECTION
%%%%%%%%%%%%%%%%%%%%%%%%%%%%%%%%%%%%%%%%%%%%%%%%%%%%%%%%%%%%%%%%%%%%%%%%%%%%%%%%%%%%%%%%%%%%%%%%%%%%%%%%%%%%%%%%%
\section{Robust GF identification}\label{S:rfi}
% Sec. intro - non-convex problem
This section presents the optimization problem and the proposed algorithm to estimate $\bbH$ and $\bbS$ under the setting described in \cref{P:problem_statement}.
Given the matrices $\bbX$, $\bbY$, and $\barbS$, we approach the robust GF identification task by means of the following non-convex optimization
\begin{alignat}{2}\label{eq:rfi_nonconvex}
    \!&\! \hbH, \hbS = \argmin_{\bbH, \bbS} && \;\; \|\bbY-\bbH\bbX\|_F^2 + \lambda d(\bbS, \barbS) + \beta \|\bbS\|_0 \nonumber \\
    \!&\! \hspace{1.2cm}\mathrm{\;\;s. \;to: } && \;\; \bbS \in \ccalS, \;\; \bbS\bbH = \bbH\bbS,
\end{alignat}
where $\mathrm{s.\;to}$ stands for $\mathrm{subject\;to}$. The first term in the objective promotes the linear input-output relation in \eqref{eq:observation_model}, encouraging the norm of $\bbW=\bbY - \bbH\bbX$ to be small. The use of the Frobenius norm is well-justified when the observation noise is Gaussian and white, but other types of noise could be accommodated by using a different norm. The second term incorporates the assumption (\textbf{AS2}) as a regularizer to obtain an estimate $\hbS$ that is related to the given GSO $\barbS$. The $\ell_0$ norm in the third term accounts for the fact of $\bbS$ being sparse. Clearly, if additional information about $\bbS$ is available, it can be incorporated into \eqref{eq:rfi_nonconvex}, either as a regularizer (e.g., a statistical prior quantifying the log-likelihood of a class of GSOs) or as a constraint that \emph{must} be satisfied (e.g., the GSO being symmetric). The latter is indeed the role of $\bbS\in\ccalS$ in \eqref{eq:observation_model}, with $\ccalS$ representing a (desired) family of GSOs such as the set of adjacency matrices with no self-loops ($\ccalS$ is the set of matrices with non-negative entries whose diagonal entries are zero) or the set of combinatorial graph Laplacians (matrices with non-positive off-diagonal entries and zero row-sum). Finally, the (key) constraint $\bbS\bbH = \bbH\bbS$ captures the fact of $\bbH$ being a polynomial of $\bbS$ and not of $\barbS$ (\textbf{AS1}). Note first that the constraint is pertinent, if $\bbH$ is a polynomial of $\bbS$, then $\bbH$ and $\bbS$ have the same eigenvectors and, as a result, their product commutes \cite{segarra2017optimal}. More importantly for the GF-identification at  hand, when the GSO is perfectly known the model $\bbH=h_0\bbI+h_1\bbS+...+h_{N-1}\bbS^{N-1}$ is linear in the unknown $\bbh$. As a result, a formulation that estimates $\bbh$ directly (as carried out in classical non-robust approaches) is well-motivated. However, when both $\bbh$ and $\bbS$ are unknown, the model $\bbH=h_0\bbI+h_1\bbS+...+h_{N-1}\bbS^{N-1}$ is highly non-linear in $\bbS$, challenging the development of a tractable solution that jointly estimates $\bbh$ and $\bbS$. Our formulation bypasses this problem by recasting the optimization variables as $\bbH$ and $\bbS$, leading to the (more tractable) bilinear constraint in \eqref{eq:rfi_nonconvex}. Nonetheless, if estimating $\bbh$ is the ultimate goal, this can be readily achieved from $\hbH$ and $\hbS$ as
\begin{equation}
    \hbh = \Big(\vvec(\bbI), \vvec(\hbS),..., \vvec(\hbS^{N-1})\Big)^\dagger \vvec(\hbH).
\end{equation}

% Formulation advantages and selected distance
The approach put forth in \eqref{eq:rfi_nonconvex} has two main advantages. 
First, while most works formulate the recovery of the GF in the spectral domain, our formulation operates in the vertex domain.
Working on the spectral domain would imply finding the Vandermonde GFT matrix $\bbPsi$. Since this matrix involves high-order polynomials of the eigenvalues of the GSO, it is also prone to numerical instability and error accumulation~\cite{djuric2018cooperative}.
Even if approaches that bypass this issue by estimating the graph-frequency response $\tbh = \bbPsi\bbh$ in lieu of $\bbh$ are adopted, the estimation would still be challenging since they require computing the eigenvectors $\bbV$, which are known to be highly sensitive to errors in the GSO (especially those associated with small eigenvalues)\cite{segarra2015stability,ceci2020graph}. On top of this, characterizing the spectral errors and incorporating those to the optimization is not a trivial task. 
%In contrast, we bypass the need to compute the spectral decomposition by postulation a formulation that works entirely in the vertex domain. 
The second advantage emanates from casting the true GSO $\bbS$ as an explicit optimization variable. As already explained, this approach is robust to error accumulation and facilitates the incorporation of the (additive) effect of the perturbations into the optimization. An additional benefit is that we obtain a denoised version (enhanced estimation) of the true GSO, which can be practically relevant in most real-world applications. 
%In contrast, we bypass the need to compute the spectral decomposition by postulation a formulation that works entirely in the vertex domain,  the graph is as useful as estimating the generative filters.
%Furthermore, disregarding the perturbations in $\barbS$ constitutes a challenge in spectral-based approaches.
%It can be rigorously shown that small perturbations in $\bbS$ can lead to significant perturbations in $\bbV^{-1}$ \cite{segarra2015stability,ceci2020graph} and, even when not large, characterizing how those perturbations translate to the eigenvectors $\bbV$ and incorporating that into the optimization is not an easy task.

In a nutshell, in the context of robust GF identification, choosing a formulation that: i) works entirely in the vertex domain, ii) considers $\bbS$ as an explicit optimization variable, and iii) codifies the GF structure via the constraint $\bbH\bbS=\bbS\bbH$, exhibits multiple advantages. However, it must be noted that the number of optimization variables is larger than in classical approaches (adding computational complexity) and that the bilinear filtering constraint $\bbH\bbS=\bbS\bbH$, while more tractable than its polynomial counterpart, is still non-convex. Alternatives to deal with these issues are discussed in later sections.

%Despite the discussed advantages of the formulation in \eqref{eq:rfi_nonconvex}, the optimization problem is non-convex and challenging to solve. A convex alternative is presented next.

\subsection{Alternating minimization for robust GF identification}
This section presents a systematic efficient approach to find an approximate solution to \eqref{eq:rfi_nonconvex}. Since the goal is to design specific algorithms, from this section onwards, we particularize the GSO distance to $d(\bbS,\barbS)=\|\bbS-\barbS\|_0$, so that, according to the discussion in Sec.~\ref{S:graph_pert}, the focus is on graph perturbations that create and destroy links. Apart from its practical relevance, the reason for choosing the $\ell_0$ norm as a distance is also motivated by its more intricate (challenging) structure. Indeed, the algorithms presented next can be easily adapted to (more tractable) distances associated with alternative perturbation models. 
Having clarified this, the main obstacle to solving \eqref{eq:rfi_nonconvex} is its lack of convexity, which emanates from two different \emph{sources}: \emph{(s1)} the $\ell_0$ norms in the objective, and \emph{(s2)} the bilinear constraint involving $\bbS$ and $\bbH$.
Next, we explain the strategy adopted to deal with them and find a solution to \eqref{eq:rfi_nonconvex} by solving a succession of convex problems.
% Dealing with non-convexity
\begin{itemize}%[leftmargin=2.5mm]
\item Regarding the $\ell_0$ norm in \emph{(s1)}, a workhorse approach is to replace it with its convex surrogate, the $\ell_1$ norm. However, it is possible to exploit more sophisticated (non-convex) alternatives that typically lead to sparser solutions.
The one chosen in this paper is to approximate the $\ell_0$ norm of a generic matrix $\bbZ \in \reals^{I \times J}$ using the logarithmic penalty
\begin{equation}\label{eq:log_sparsity}
     \| \bbZ \|_0 \approx r_\delta(\bbZ) := \sum_{i=1}^I\sum_{j=1}^J \log(|Z_{ij}| + \delta),
\end{equation}
where $\delta$ is a small positive constant~\cite{candes2008enhancing}.
The non-convexity of the logarithm can be handled efficiently by relying on a majorization-minimization approach (MM)~\cite{sun2016majorization}, which considers an iterative linear approximation leading to an iterative re-weighted $\ell_1$ norm.
It is worth noting that, since we will consider an iterative algorithm to deal with the bilinearity of \eqref{eq:rfi_nonconvex}, the iterative nature of the re-weighted $\ell_1$ norm will not impose a significant computational burden.
Details on the exact form of this sparse regularizer will be provided soon, when describing the estimation of $\bbS$. 
\item To deal with the bilinear terms in \emph{(s2)}, we adopt an alternating optimization approach \cite{gorski2007biconvex} resulting in an iterative algorithm where the optimization variables $\bbH$ and $\bbS$ are updated in \emph{two separate iterative steps}.
At each step, we optimize over one of the optimization variables with the other remaining fixed, resulting in two simpler problems that can be solved efficiently. The details about the specific steps will be provided shortly.
\end{itemize}

% New objective function
Taking into account these considerations, the first task to implement our approach is to rewrite the problem in \eqref{eq:rfi_nonconvex} as
% \begin{alignat}{2}\label{eq:rfi_nonconvex_rew}
%     \!&\!\min_{\bbS\in \ccalS, \bbH} && \;\;\|\bbY-\bbH\bbX\|_F^2 + \lambda r_{\delta_1}(\bbS - \barbS) + \beta r_{\delta_2}(\bbS)   \nonumber     \\
%     \!&\! && \;\;+ \gamma \|\bbS\bbH - \bbH\bbS\|_F^2, %\nonumber \\
%     %\!&\! \mathrm{\;\;s. \;to: } && \;\;\bbS \in \ccalS,
% \end{alignat}
\begin{equation}\label{eq:rfi_nonconvex_rew}
    \min_{\bbS\in \ccalS, \bbH} \|\bbY \!-\!\bbH\bbX\|_F^2 \!+\! \lambda r_{\delta_1}\!(\bbS \!-\! \barbS) \!+\! \beta r_{\delta_2}\!(\bbS) \!+\! \gamma \|\bbS\bbH \!-\! \bbH\bbS\|_F^2,
\end{equation}
where we recall that $r_{\delta}(\cdot)$ was introduced in \eqref{eq:log_sparsity}. Note that: i) the logarithmic penalty has also been used to promote sparsity in the term $\bbS - \barbS$ since we selected the $\ell_0$ norm as the distance between $\bbS$ and $\barbS$, and ii) the constraint $\bbS\bbH = \bbH\bbS$ was relaxed and rewritten as a regularizer, a formulation more amenable to an alternating optimization approach.

% Iterative algorithm
The next task is to solve \eqref{eq:rfi_nonconvex_rew} by means of an iterative algorithm that blends techniques from alternating optimization and MM approaches.
Specifically, for a maximum of $t_{max}$ iterations, we run the following two steps at each iteration $t = 0,...,t_{max}-1$. 

\vspace{2mm}
\noindent \textbf{Step 1: GF Identification.}
We estimate the block of $N^2$ variables collected in $\bbH$ while the current estimate of the GSO, denoted as $\bbS^{(t)}$, remains fixed.
This results in the convex optimization problem
\begin{alignat}{2}\label{eq:step1_filterid}
\!\!&\bbH^{(t+1)} =  \arg \min_{\bbH} && \|\bbY\!-\!\bbH\bbX\|_F^2 \!+\! \gamma \|\bbS^{(t)}\bbH \!-\! \bbH\bbS^{(t)}\|_F^2,
\end{alignat}
an LS minimization whose closed-form solution is 
\begin{align}\label{eq:step1_closed_form}
    \vvec(\bbH^{(t+1)}) \!= &\big(\bbX\bbX^\top \!\!\otimes\! \bbI \!+\! \gamma (\bbS\bbS^\top \!\!\oplus\! \bbS^\top\bbS \!-\! \bbS^\top \!\!\otimes\! \bbS^\top \!\!-\! \bbS \!\otimes\!  \bbS)\big)^{-1}\nonumber\\
    &\times (\bbX \!\otimes \bbI)\vvec(\bbY).
\end{align} 
Here, $\otimes$ is the Kronecker product, $\oplus$ is the Kronecker sum, and $\bbI$ is the identity matrix of size $N \times N$.
Also note that \eqref{eq:step1_closed_form} omitted the iteration superscript in $\bbS^{(t)}$ to alleviate notation. 

\vspace{2mm}
\noindent \textbf{Step 2: Graph Denoising.}
Following an MM scheme, we optimize an upper bound of \eqref{eq:rfi_nonconvex_rew} where the logarithmic penalties are linearized.
Then, we estimate the block of $N^2$ variables collected in $\bbS$ while the current estimate of the GF $\bbH^{(t+1)}$ remains fixed.
This yields
\begin{align}\label{eq:step2_graph_denoising}
    \bbS^{(t+1)} =  \arg\min_{\bbS \in \ccalS} &\sum_{i=1}^N\sum_{j=1}^N \big(\lambda \bar{\Omega}_{ij}^{(t)}|S_{ij}-\bar{S}_{ij}|+\beta \Omega_{ij}^{(t)}|S_{ij}|\big) \nonumber \\
     &+ \gamma \|\bbS\bbH^{(t+1)} - \bbH^{(t+1)}\bbS\|_F^2,
\end{align}
where $\bar{\bbOmega}^{(t)}$ and $\bbOmega^{(t)}$ are computed in an entry-wise fashion based on the GSO estimate from the previous iteration as
\begin{align}\label{eq:log_weights}
    &\bar{\Omega}_{ij}^{(t)} = \frac{1}{|S_{ij}^{(t)} - \bar{S}_{ij}| + \delta_1},
    &\Omega_{ij}^{(t)} = \frac{1}{|S_{ij}^{(t)}| + \delta_2}.
\end{align}

%%%%%%%%%%  ALGORITHM  %%%%%%%%%%
\begin{algorithm}[tb]
\footnotesize
\SetKwInput{Input}{Input}
\SetKwInOut{Output}{Output}
\Input{$\bbX$, $\bbY$, $\barbS$}
\Output{$\hbH$, $\hbS$.}
\SetAlgoLined
Initialize $\bbS^{(0)}$ as $\bbS^{(0)} = \barbS$. \\
\For{$t=0$ \KwTo $t_{max}-1$}{
    Compute $\bbH^{(t+1)}$ by solving \eqref{eq:step1_closed_form} fixing $\bbS^{(t)}$. \\
    Update $\bbOmega^{(t)}$ and $\bar{\bbOmega}^{(t)}$ as in \eqref{eq:log_weights}. \\
    Compute $\bbS^{(t+1)}$ by solving \eqref{eq:step2_graph_denoising} using $\bbH^{(t+1)}$, $\bbOmega^{(t)}$, and $\bar{\bbOmega}^{(t)}$.
    \\
}
$\hbH = \bbH^{(t_{max})},\; \hbS = \bbS^{(t_{max})}$.
\caption{\small Robust GF identification with graph denoising.}
\label{A:rfi_alg}
%\vspace{-.2cm}
\end{algorithm}
%%%%%%%%%%  END ALGORITHM  %%%%%%%%%%

% Comments on the algorithm - convergence
The overall alternating algorithm is summarized in Alg.~\ref{A:rfi_alg}, where a fixed number of iterations is considered.
The algorithm starts by initializing the GSO as  $\bbS^{(0)} = \barbS$ (although other options could also be appropriate), and then, it iterates between Steps 1 and 2 for a fixed number of epochs (or until some stopping criterion is met).
In this regard, a key feature of the algorithm is that it is guaranteed to converge to a stationary point, as is formally stated next.

\begin{theorem}\label{thm1}
    Denote as $f(\bbH,\bbS)$ the objective function in \eqref{eq:rfi_nonconvex_rew}, and let $\ccalZ^*$ be the set of stationary points of $f$.
    Let $\bbz^{(t)}=[\vvec(\bbH^{\!(t)})^{\!\top}\!, \vvec(\bbS^{\!(t)})^{\!\top}\!]^{\!\top}$ represent the solution provided by the iterative algorithm \eqref{eq:step1_closed_form}-\eqref{eq:step2_graph_denoising} after $t$ iterations. Assuming that i)~the GSO does not have repeated eigenvalues and ii)~every row of $\tbX\!=\!\bbV^{\!-1}\bbX$ has at least one nonzero entry, then $\bbz^{\!(t)}$ converges to a stationary point of $f$ as $t$ goes to infinity, i.e.,
        \begin{equation}
            \lim_{t\to\infty} \mathsf{d}(\bbz^{(t)}~|\ccalZ^*) = 0, \nonumber
        \end{equation}
    with $\mathsf{d}(\bbz~|\ccalZ^*) := \min_{\bbz^* \in \ccalZ^*} \|\bbz-\bbz^*\|_2$.
\end{theorem}

The proof relies on the convergence results shown in~\cite[Th. 1b]{hong2015unified} and the details are provided in App.~A.
Note that the convergence of the algorithm was not self-evident since the original optimization problem in \eqref{eq:rfi_nonconvex_rew} is non-convex and Step 2 is minimizing an upper-bound of the original objective function. The sufficient conditions in i) and ii) guarantee that every graph frequency is excited so that the GF is identifiable and \eqref{eq:step1_filterid} has a unique solution, which is a requirement for convergence (see Prop.~\ref{thm3} in App.~A for details). Clearly, condition ii) is fulfilled even for $M=1$ if all the entries of the vector $\tbx=\bbV^{-1}\bbx$ are nonzero. Alternatively, when $M>1$ and ii) is satisfied, condition i) can be relaxed.

% Sensitivity of gamma
Another relevant element in the proposed algorithm is the weight $\gamma$.
If $\gamma$ is set to a value that is too large, the GF estimated in the first iteration $\bbH^{(1)}$ will be an (almost exact) polynomial of $\bar{\bbS}$ so that the algorithm will converge quickly to the same solution as that of the non-robust design % that assumes that the true GSO is $\barbS$
[cf. \eqref{eq:rfi_nonconvex} with $\bbS=\bar{\bbS}$].
On the other hand, if $\gamma$ is too close to zero the two problems decouple and the solution converges quickly to that of the two separated problems [cf. \eqref{eq:step1_filterid} and \eqref{eq:step2_graph_denoising} with $\gamma=0$].
As a result, the value of the parameter must be chosen carefully. In this context, schemes that start with a small $\gamma$ to encourage the exploration during the warm-up phase, and then increase $\gamma$ as the iteration index grows to guarantee that the final $\hbH$ is a polynomial of $\hbS$ are a suitable alternative for the setup at hand.

Finally, one drawback of the proposed robust GF identification algorithm is that the optimization problems in \eqref{eq:step1_filterid} and \eqref{eq:step2_graph_denoising} may be slow when dealing with large graphs.
However, we will mitigate this issue by introducing an efficient implementation that reduces the computational complexity of the overall algorithm (see Sec.~\ref{S:efficient_rfi}).

\subsection{Leveraging stationary observations}\label{S:stationary_observations}
The alternating convex approximation in Alg.~\ref{A:rfi_alg} exploits the fact that $\bbX$ and $\bbY$ are linearly related via $\bbH$, which is a polynomial of $\bbS$.
However, in setups where the perturbations in $\barbS$ are very large, obtaining accurate estimates of $\bbS$ and $\bbh$ from $\hbH$ may still be challenging. One alternative to overcome this issue is to leverage the additional structure potentially present in our data.
Indeed, as detailed in the introduction, it is common to consider setups where the signals exhibit additional properties depending on the supporting graph, with notable examples including graph-bandlimited signals~\cite{shuman2013emerging,sandryhaila2014discrete}, diffused sparse graph signals~\cite{segarra2017blind,rey2019sampling}, or graph stationary signals~\cite{marques2017stationary,shafipour2020online,buciulea2022learning}. Clearly, incorporating such additional information into the optimization problem would enhance its estimation performance.

This section explores this path, restricting our attention to the case where the observed signals are stationary on $\ccalG$. The motivation for this decision is that, due to the tight connection between graph-stationary signals and GFs (see Sec.~\ref{S:preliminaries}), the formulation in \eqref{eq:rfi_nonconvex_rew} and Alg.~\ref{A:rfi_alg} require relatively minor modifications to incorporate the assumption of $\bbX$ and $\bbY$ being stationary on $\bbS$, leaving the incorporation of additional signal models as future work. To formulate the updated problem, recall that the covariance matrix of a stationary graph signal can be expressed as a polynomial of the GSO (see Sec.~\ref{S:preliminaries}). Therefore, incorporating stationarity calls for modifying \eqref{eq:rfi_nonconvex_rew} as
\begin{alignat}{2}\label{eq:rfi_nonconvex_st}
    \!&\!\min_{\bbS \in \ccalS, \bbH} && \;\!\|\bbY\!-\!\bbH\bbX\|_F^2 + \lambda r_{\delta_1}\!(\bbS \!-\! \barbS) + \beta r_{\delta_2}\!(\bbS)+ \gamma \|\bbS\bbH \!-\! \bbH\bbS\|_F^2   \nonumber     \\
  %  \!&\! && \;+ \gamma \|\bbS\bbH - \bbH\bbS\|_F^2 \nonumber \\
    \!&\! \mathrm{\;\;s. \;to: } && \; \|\bbC_\bby\bbS\!-\!\bbS\bbC_\bby\!\|_F^2\!\leq\! \epsilon_\bby,\,\!\|\bbC_\bbx\bbS\!-\!\bbS\bbC_\bbx\!\|_F^2\!\leq\!\epsilon_\bbx,
\end{alignat}
where $\bbC_\bby$ and $\bbC_\bbx$ denote the covariance matrices of $\bbY$ and $\bbX$, respectively.
If the covariances are perfectly known, then the corresponding parameters $\epsilon_\bby$ and $\epsilon_\bbx$ are set to zero.
Alternatively, if the $\bbC_\bby$ and $\bbC_\bbx$ are the sample estimates of the true covariances, then the values of $\epsilon_\bby$ and $\epsilon_\bbx$ must be selected based on the quality of the estimators (accounting, e.g., for the number of available observations $M$).

The constraints in \eqref{eq:rfi_nonconvex_st} capture the graph-stationarity assumption by promoting the commutativity with the true GSO.
Therefore, such constraints are considered in the graph denoising step [cf. \eqref{eq:step2_graph_denoising}].
In addition, since $\bbC_\bby$, $\bbC_\bbx$ and $\bbH$ are all polynomials of $\bbS$, the equalities $\bbC_\bby \bbH=\bbH\bbC_\bby$ and $\bbC_\bbx \bbH=\bbH\bbC_\bbx$ must hold as well, so it is also possible to augment the GF identification step [cf. \eqref{eq:step1_filterid}] with the corresponding constraints.
While in the interest of brevity, we do not spell out all the possible formulations here, the impact of several of these alternatives is numerically analyzed in Sec.~\ref{S:experiments}. Finally, it is important to note that, since the stationarity constraints are quadratic and convex, the convergence described in \cref{thm1} also holds true for the iterative algorithm associated with \eqref{eq:rfi_nonconvex_st}.

%%%%%%%%%%%%%%%%%%%%%%%%%%%%%%%%%%%%%%%%%%%%%%%%%%%%%%%%%%%%%%%%%%%%%%%%%%%%%%%%%%%%%%%%%%%%%%%%%%%%%%%%%%%%%%%%%
%SECTION: MULTIPLE FILTERS
%%%%%%%%%%%%%%%%%%%%%%%%%%%%%%%%%%%%%%%%%%%%%%%%%%%%%%%%%%%%%%%%%%%%%%%%%%%%%%%%%%%%%%%%%%%%%%%%%%%%%%%%%%%%%%%%%
\section{Joint robust identification of multiple GFs}\label{S:rfi_joint}
% Intro
In Sec.~\ref{S:rfi}, we approached the problem of identifying a single GF $\bbH$ defined over a single graph $\ccalG$.
However, in a variety of situations we encounter multiple processes (signals) over the same graph $\ccalG$.
Consider for example a network of weather stations measuring the temperature, humidity, and wind speed.
Each of these measurements corresponds to observations of a different process, all of them taking place over a common graph.
Intuitively, since all the GFs are related by the underlying graph $\ccalG$, we propose a \emph{joint} GF identification approach that exploits this relationship to enhance the quality of the estimation.
We focus first on the case where the input-output signals associated with each GF (graph process) are observed separately. Later in the section, we address a slightly more involved case where the GFs model the (AR) dynamics of a time-varying graph signal and, as a result, the observed signals are intertwined.

% New definitions
Consider a set of $K$ unknown GFs $\{\bbH_k\}_{k=1}^K$, all represented by $N \times N$ matrices and defined over the  graph $\ccalG$.
To be consistent with \cref{P:problem_statement}, we assume that: i) the true $\bbS$ is unknown and only the perturbed version $\barbS$ is available; ii) all $\bbH_k$ are polynomials of the \emph{same} GSO $\bbS$; and iii) for each $k$, matrices $\bbX_k \in \reals^{N \times M_k}$ and $\bbY_k \in \reals^{N \times M_k}$ collect the observed input and output graph signals and are related via
\begin{equation}\label{eq:observation_model_multi}
    \bbY_k = \bbH_k\bbX_k + \bbW_k,
\end{equation}
with $\bbH_k=\sum_{r=0}^{N-1}h_{r,k}\bbS^r$ and $\bbW_k$ being a white random matrix capturing observation noise and model inaccuracies.
Then, we aim at estimating the GFs $\{\bbH_k\}_{k=1}^K$ in a joint fashion while taking into account the inaccuracies in the topology of $\ccalG$.
This is summarized in the following problem statement.

\begin{problem}\label{P:multiple_filter_ir}
    Let $\ccalG$ be a graph with $N$ nodes, let $\bbS \in \reals^{N \times N}$ be the true (unknown) GSO associated with $\ccalG$, and let $\barbS \in \reals^{N \times N}$ be the perturbed (observed) GSO.
    Moreover, let $\bbX_k\in \reals^{N \times M_k}$ and $\bbY_k \in \reals^{N \times M_k}$ be the matrices collecting the $M_k$ observed input and output graph signals associated with $k=1,...,K$ network processes, all defined over $\ccalG$ and adhering to the model in \eqref{eq:observation_model_multi}.
    Our goal is to use $\{\bbX_k\!\}_{k=1}^K$, $\{\bbY_k\!\}_{k=1}^K$, and $\barbS$ to learn the $K$ GFs $\{\bbH_k\}_{k=1}^K$ that best fit the data, along with an enhanced estimation of $\bbS$.
    To that end, we make the following assumptions: \\
    (\textbf{AS2}) $\bbS$ and $\barbS$ are close according to some metric $d(\bbS,\barbS)$, i.e., the observed perturbations are ``small'' in some sense. \\
    (\textbf{AS3}) Every $\bbH_k$ is a polynomial of $\bbS$. \\
    \vspace{-0.3cm}
\end{problem}

Assumption (\textbf{AS2}), which was also considered in \cref{P:problem_statement}, promotes the tractability of the problem by ensuring that $\bbS$ and $\barbS$ are related. As discussed in Sec.~\ref{S:graph_pert}, the distance function $d(\cdot, \cdot)$ must be selected depending on the perturbation model at hand. 
(\textbf{AS3}) captures the key fact that all the matrices $\bbH_k$ are GFs of the \emph{same} GSO, establishing a link that can be leveraged via a joint estimation (optimization) of the $K$ GFs. 
Implementing an approach similar to that in Sec.~\ref{S:rfi} (i.e., working on the vertex domain, considering the true GSO as an explicit optimization variable, accounting for the GF structure via a commutativity constraint, and assuming that the graph perturbations create and destroy links), the multi-filter counterpart to \eqref{eq:rfi_nonconvex_rew} that codifies \cref{P:multiple_filter_ir} is
\begin{alignat}{2}\label{eq:joint_rfi_noncvx_rew}
    \!&\! \min_{\bbS\in \ccalS, \{\bbH_k\}_{k=1}^K} && \sum_{k=1}^K  \alpha_k\|\bbY_k-\!\bbH_k\bbX_k\|_F^2 \!+\! \lambda r_{\delta_1}(\bbS - \barbS) \nonumber \\
    \!&\! &&  +\beta r_{\delta_2}(\bbS) + \sum_{k=1}^K \gamma \|\bbS\bbH_k \!\!-\! \bbH_k\bbS\|_F^2. %\nonumber \\
    %\!&\! \mathrm{\;\;s. \;to: } && \bbS \in \ccalS.
\end{alignat}
Ideally, the value of the positive weight $\alpha_k$ must be selected based on the norm of $\bbW_k$ (e.g., prior information on the noise level and the number of signal pairs $M_k$). If none is available, then $\alpha_k=1$ for all $k$.
Equally important, the fact of pursuing a joint optimization implies that each $\bbH_k$ contributes with a regularization term $\|\bbS\bbH_k - \bbH_k\bbS\|_F^2$ promoting the commutativity of the $k$-th GF with the \emph{single} $\bbS$. Intuitively, having the same $\bbS$ in all these terms couples the optimization across $k$ and contributes to reduce the uncertainty over $\bbS$, leading to enhanced estimates of both $\bbS$ and $\{\bbH_k\}_{k=1}^K$.
As a result, the joint GF identification approach is expected to provide better results than estimating each $\bbH_k$ separately by solving $K$ instances of \eqref{eq:rfi_nonconvex_rew}.
We validate this hypothesis numerically via the experiments in Sec.~\ref{S:experiments}.

Following a motivation similar to that in the previous section, we deal with the non-convex minimization in \eqref{eq:joint_rfi_noncvx_rew} designing an alternating optimization algorithm that breaks the bilinear terms $\bbS\bbH_k$ and $\bbH_k\bbS$, and approximates the logarithmic terms with a linear upper-bound. The resulting algorithm solves iteratively the following two subproblems for $t=1,...,t_{\max}$ iterations.

\vspace{2mm}
\noindent \textbf{Step 1: Multiple GF Identification.}
Given the current estimate $\bbS^{(t)}$, we solve the optimization problem in \eqref{eq:joint_rfi_noncvx_rew} with respect to each $\bbH^{(k)}$.
This yields
%\begin{alignat}{2}\label{eq:joint_filterid}
%    \bbH_k^{(t+1)} \!\!\!=\! &\argmin_{\bbH_k} ~\alpha_k \Big\| \bbY_k \!\!-\!\! \bbH_k\bbX_k -\sum_{k'<k}\!\bbH_{k'}^{(t+1)}\bbX_{k'} \Big.\nonumber \\
%    &~~-\!\Big.\!\!\sum_{k'>k}\!\bbH_{k'}^{(t)}\bbX_{k'}\!\Big\|_F^2  \!+ \gamma \Big\| \bbS^{(t)}\bbH_k \!\!-\!\! \bbH_k\bbS^{(t)} \Big\|_F^2,
%\end{alignat}
%whose closed-form solution can be found using \eqref{eq:step1_closed_form} replacing $\gamma$ with $\gamma/\alpha_k$, $\bbX$ with $\bbX_k$, and $\bbY$ with $\bbY_k-\sum_{k'< k}\bbH_{k'}^{(t+1)}\bbX_{k'}-\sum_{k'> k}\bbH_{k'}^{(t)}\bbX_{k'}$.
%Note that, since we were using an alternating minimization scheme, \eqref{eq:joint_filterid} approaches estimating each $\bbH_k^{(t+1)}$ separately from the other filters, solving $K$ decoupled LS problems (each with $N^2$ unknowns). Furthermore, if multiple processors were available, a modified version %of \eqref{eq:joint_filterid} replacing $\bbH_{k'}^{(t+1)}$ with $\bbH_{k'}^{(t)}$ would render the filter-identification step amenable to a parallel implementation. Alternatively, we could solve for all the $\{\bbH_k\}_{k=1}^K$ simultaneously, solving a single LS problem with $KN^2$ variables. 
\begin{alignat}{2}\label{eq:joint_filterid}
    \!\!\bbH_k^{(t+1)\!} \!\!=\!\argmin_{\bbH_k} \alpha_k \|\! \bbY_k \!\!-\!\! \bbH_k\bbX_k\!\|_{\!F}^{\!2}  \!\!+\! \gamma \|\! \bbS^{\!(t)\!}\bbH_k \!\!-\!\! \bbH_k\bbS^{\!(t)} \!\|_{\!F}^{\!2},\!
\end{alignat}
whose closed-form solution can be found using \eqref{eq:step1_closed_form} replacing $\gamma$ with $\gamma/\alpha_k$, $\bbX$ with $\bbX_k$, and $\bbY$ with $\bbY_k$.
Note that since the only coupling across GFs is via the GSO, \eqref{eq:joint_filterid} estimates each $\bbH_k^{(t+1)}$ separately from the other GFs, solving $K$ LS problems (each with $N^2$ unknowns). Furthermore, if multiple processors are available, \eqref{eq:joint_filterid} can be run in parallel across $k$.

\vspace{2mm}
\noindent \textbf{Step 2: Graph Denoising.}
Given the current estimates of the GFs $\{\bbH_k^{(t+1)}\}_{k=1}^K$, we follow an MM scheme that, minimizing a linear upper-bound of the logarithmic penalties, yields the estimate of the GSO via
\begin{align}\label{eq:joint_graph_denoising}
    \bbS^{(t+1)} =  \argmin_{\bbS \in \ccalS} &\sum_{ij=1}^N \big( \lambda \bar{\Omega}_{ij}^{(t)}|S_{ij}-\bar{S}_{ij}|+\beta \Omega_{ij}^{(t)}|S_{ij}| \big) \nonumber \\
     &+ \sum_{k=1}^K\gamma \|\bbS\bbH_k^{(t+1)} - \bbH_k^{(t+1)}\bbS\|_F^2,
\end{align}
where $\bbOmega$ and $\bar{\bbOmega}$ are obtained as in \eqref{eq:log_weights}.

The solution to \cref{P:multiple_filter_ir} is simply given by $\hbS = \bbS^{(t_{max})}$ and $\hbH_k = \bbH_k^{(t_{max})}$ for every $k$. Similar to \eqref{eq:rfi_nonconvex_rew}, convergence to a stationary point of \eqref{eq:joint_rfi_noncvx_rew} is guaranteed, as formally stated next.

\begin{corollary}\label{thm2}
    Denote as $f(\{\bbH_k\}_{k=1}^K, \bbS)$ the objective function in \eqref{eq:joint_rfi_noncvx_rew}.
    If $\bbz^{(t)} = [\vvec(\bbH_1^{(t)})^\top,...,\vvec(\bbH_K^{(t)})^\top, \vvec(\bbS)^\top]^\top$ represents the solution provided by the iterative algorithm \eqref{eq:joint_filterid}-\eqref{eq:joint_graph_denoising} after $t$ iterations and every $\bbX_k$ excites all graph frequencies, then $\bbz^{(t)}$ converges to a stationary point of $f$ as the number of iterations $t$ goes to infinity.
\end{corollary}
The key to prove \cref{thm1}, which established the convergence to a stationary point for the robust estimation of a single GF, was to show that the optimization problem in \eqref{eq:rfi_nonconvex_rew} and the proposed algorithm satisfied the conditions in \cite[Th. 1b]{hong2015unified}. The formulation we put forth for the multi-filter case resembles closely that of the single-filter case, and, as a result, it is not difficult to show that those conditions also hold true for the problem in \eqref{eq:joint_rfi_noncvx_rew} (see App.~A for details). 

The discussion and formulations in Sec.~\ref{S:stationary_observations} dealing with incorporating additional information about the input-output signals into the optimization are also pertinent for the setup in this section. The details of such a formulation are omitted for brevity, but it will be explored in the experimental section.

\subsection{Joint GF identification for time series}
A slightly different, practically relevant, setup where multiple GFs need to be estimated is that of graph-based multivariate time series.
In that setup, each variable is associated with a node of the graph and the multiple graph-signal observations correspond to different instants of a time-varying graph signal. AR and moving-average (MA) modeling of time series has a long tradition, with common approaches to decrease the degrees of freedom %of the multivariate model
including limiting the memory of the series and assuming that matrices of coefficients relating different time instants are low rank \cite{reinsel2003elements}.
%In the context of graph signals and network processes, a natural approach is to assume that the matrices of coefficients are accurately modeled as GFs, all defined over the same graph \cite{mei2017causal,isufi19timeSeriesVarma}.
In the context of graph signals and network processes, a natural approach is to constrain the matrices of coefficients to be GFs, all defined over the same graph \cite{mei2017causal,isufi19timeSeriesVarma}.
This section introduces a variation of the problem in \eqref{eq:joint_rfi_noncvx_rew} tailored to this setup. 

To introduce the multiple-graph identification problem formally, let $\bbX_\kappa$ and $\bbY_\kappa$ denote a collection of $M_\kappa$ graph signals corresponding to measurements of a network process for $\kappa=1,...,\kappa_{max}$ time instants.
%Suppose now that the network process $\bbY_\kappa$ can be accurately modeled by an AR dynamics with memory $K$. This implies that, at every instant $\kappa$, the observations $\bbY_\kappa$ satisfy the equation
Suppose now that $\bbY_\kappa$ can be accurately modeled by an AR dynamics with memory $K$ so, at every instant $\kappa$, the observations $\bbY_\kappa$ satisfy the equation
\begin{equation}\label{eq:ar_observation}
    \bbY_\kappa = \sum_{k=1}^K \bbH_k\bbY_{\kappa-k} + \bbX_\kappa,\;\mathrm{with}\;\bbH_k=\sum_{r=0}^{N-1}h_{r,k}\bbS^r, 
\end{equation}
where $\bbX_\kappa$ is the exogenous input, and the GF $\bbH_k$ models the influence that the signal observations from the time instant $\kappa-k$ exert on the (current) signal at time $\kappa$.
%Although we only consider the influence of the exogenous input $\bbX_\kappa$ at the current time $\kappa$ for simplicity, modifying \eqref{eq:ar_observation} to account for additional time instants of $\bbX$ is straightforward.

Suppose now that: i) we have access to an estimated (imperfect) graph $\barbS$, ii) the value of the graph signals at different time instants is available, and iii) our goal is to estimate the set of matrices (GFs) $\{\bbH_k\}_{k=1}^K$ in \eqref{eq:ar_observation} that describe the dynamics of the multivariate time series. This can be accomplished as
\begin{alignat}{2}\label{eq:joint_rfi_noncvx_rew_ar}
    \!&\! \min_{\bbS\in \ccalS, \{\bbH_k\}_{k=1}^K}  \sum_{\kappa=K+1}^{\kappa_{max}}\Big\|\bbY_\kappa-\bbX_\kappa -\!\sum_{k=1}^K\bbH_k\bbY_{\kappa-k}\Big\|_F^2   \nonumber \\
    \!&\! \hspace{0.6cm} +\! \lambda r_{\delta_1}(\bbS - \barbS) \!+\! \beta r_{\delta_2}(\bbS) \!+\! \sum_{k=1}^K \gamma \|\bbS\bbH_k \!\!-\! \bbH_k\bbS\|_F^2. \!
    %\!&\! \mathrm{\;\;s. \;to: } && \bbS \in \ccalS.
\end{alignat}
The main difference relative to \eqref{eq:joint_rfi_noncvx_rew} is in the first term, which accounts for the new observation model [cf. \eqref{eq:observation_model_multi} vs. \eqref{eq:ar_observation}]. Note that we assume that the exogenous input $\bbX_\kappa$ is observed. If that were not the case, it would suffice to remove $\bbX_\kappa$ from the objective (possibly updating the Frobenius norm in case statistical knowledge about $\bbX_\kappa$ were available). Albeit the differences, the problem in \eqref{eq:joint_rfi_noncvx_rew_ar} is closely related to \eqref{eq:joint_rfi_noncvx_rew}, with the sources of non-convexities being the same. As a result, we approach its solution with a modified version of Alg.~\ref{A:rfi_alg} which, at each iteration $t$, runs two steps. In the first one, we estimate each of the $K$ GFs by solving
\begin{alignat}{2}\label{eq:joint_filterid_time}
    \!& \bbH_k^{(t+1)} \!&&\!=\!\argmin_{\bbH_k}  \!\!\!\sum_{\kappa=K+1}^{\kappa_{max}}\!\!     \Big\|\bbY_\kappa\!-\!\bbX_\kappa\!-\!\bbH_k\bbY_{\kappa-k}\!-\!\!\!\sum_{k'<k}\!\!\bbH_{k'}^{(t+1)}\bbY_{\kappa-k'} \nonumber\\
            \!\!& \! &&\! \!-\!\! \sum_{k'<k}\!\!\bbH_{k'}^{(t)}\bbY_{\kappa-k'} \Big\|_{\!F}^{\!2}  +\!\! \sum_{k=1}^K\!\gamma \Big\| \bbS^{(t)}\bbH_k - \bbH_k\bbS^{(t)} \Big\|_{\!F}^{\!2},\!\!
\end{alignat}
which is different from the previous GF identification step [cf. \eqref{eq:joint_filterid}]. In contrast, the graph-denoising step in \eqref{eq:joint_graph_denoising} remains the same. Note that \eqref{eq:joint_filterid_time} updates each GF separately in a cyclic way by solving an LS problem with $N^2$ unknowns. Alternative implementations include using $\bbH_{k'}^{(t)}$ in lieu of $\bbH_{k'}^{(t+1)}$ for all $k'< k$ (so that a parallel implementation is enabled) as well as considering a single LS problem with $KN^2$ unknowns. 
%In other words, the joint GF identification algorithm for time series is an iterative approach that solves \eqref{eq:joint_filterid_time} and \eqref{eq:joint_graph_denoising} at each iteration.

Finally, it is worth emphasizing that the formulation introduced in this section can be used as a starting point to design more general robust schemes for multivariate time series defined over a graph. Dealing with both AR and MA matrices, assuming that the memory of the system is not known, having only partial/statistical information on the exogenous input, and observing the signals at only a subset of nodes are all examples of setups of interest. Since our goal in this section was to demonstrate the relevance of a robust multiple GF formulation in the context of multivariate time series, to facilitate exposition we restricted our discussion to the relatively simple case in \eqref{eq:ar_observation},  but many other setups (including those previously listed) will be subject of our future work.

%%%%%%%%%%%%%%%%%%%%%%%%%%%%%%%%%%%%%%%%%%%%%%%%%%%%%%%%%%%%%%%%%%%%%%%%%%%%%%%%%%%%%%%%%%%%%%%%%%%%%%%%%%%%%%%%%
%SECTION: FAST ALGORITHM
%%%%%%%%%%%%%%%%%%%%%%%%%%%%%%%%%%%%%%%%%%%%%%%%%%%%%%%%%%%%%%%%%%%%%%%%%%%%%%%%%%%%%%%%%%%%%%%%%%%%%%%%%%%%%%%%%
\section{Efficient implementation of the robust GF identification algorithm}\label{S:efficient_rfi}
The algorithms proposed up to this point are able to find a solution to the robust GF identification problem in polynomial time. However, their computational complexity scales with the number of nodes as $N^7$. To facilitate the deployment in setups where $N$ is large, this section puts forth an efficient implementation that reduces the number of operations.

The new algorithm %(labeled as Alg.~\ref{A:efficient_rfi_alg}, and whose pseudocode can be found in the following page) 
(summarized in Alg.~\ref{A:efficient_rfi_alg} ) preserves the core structure of Alg.~\ref{A:rfi_alg}, with an outer loop that, at each iteration, runs two steps: one involving the estimation of the GF(s) and another one dealing with the estimation of the GSO.
The main difference is that now, instead of finding the exact solution to those two problems, we obtain an approximate solution. While the details, which are step-dependent, will be specified in the next paragraphs, the overall idea is that for each of the steps we run a few simple (gradient/proximal) iterations. Although Alg.~\ref{A:efficient_rfi_alg} involves two nested loops, the complexity of the problems in the inner loop is cut down significantly, so that the overall computational overhead is reduced. 

To be specific, we describe next the two steps that, at each iteration of the outer loop $t \!=\! 0,...,t_{max}\!-\!1$, Alg.~\ref{A:efficient_rfi_alg} runs.

\vspace{2mm}
\noindent \textbf{Step 1: Efficient GF Identification.}
Solving the GF-identification step with the closed-form solution presented in \eqref{eq:step1_closed_form} involves inverting a matrix of size $N^2 \times N^2$, which requires $\ccalO(N^6)$ operations.
To explain our alternative implementation, let $f_1(\bbH|\bbS^{(t)})$ denote the objective function in \eqref{eq:step1_filterid}. Since $f_1$ is strictly convex and smooth, it can be efficiently optimized using a gradient descent approach~\cite{boyd2004convex}.

To that end, for each iteration $t$ of the outer loop, we define the inner iteration index  $\tau$ as well as the sequence of variables $\bchkH^{(\tau)}$ with $\tau = 0,...,\tau_{max_1}$, which is initialized as $\bchkH^{(0)} = \bbH^{(t)}$.
With this notation at hand, at each iteration $\tau=0,...,\tau_{max_1}-1$ of the inner loop, we update $\bchkH^{(\tau+1)}$ via
\begin{equation}
    \bchkH^{(\tau+1)} = \bchkH^{(\tau)} - \mu \nabla f_1(\bchkH^{(\tau)}|\bbS^{(t)}).
\end{equation}
Here, $\mu > 0$ is the step size and $\nabla f_1$ denotes the gradient of $f_1$ with respect to $\bbH$, which is given by
\begin{alignat}{2}\
 &\! \nabla f_1\!(\bbH|\bbS^{(t)}) \!= \!2\Big(\bbH\bbX\bbX^\top \!\!-\!\! \bbY\bbX^\top \Big)\!\!\nonumber\\
    &\hspace{0.2cm}+\!2\gamma\Big( \!\bbS^{(t)^{ \!\top}}  \!\!(\bbS^{(t)}\bbH \!-\! \bbH\bbS^{(t)}) \!-\! (\bbS^{(t)}\bbH \!-\! \bbH\bbS^{(t)})\bbS^{(t)^\top} \!\Big).
\end{alignat}
When the $\tau_{max_1}$ gradient updates are computed, we conclude the GF-identification step by setting $\bbH^{(t+1)} = \bchkH^{(\tau_{max_1})}$.

Since each gradient calculation involves the multiplication of $N \times N$ matrices, the resultant computational complexity is $\ccalO(\tau_{max_1}N^3)$, which may go down to $\ccalO(\tau_{max_1}N^{2.4})$ if an efficient multiplication algorithm is employed~\cite{coppersmith1987matrix}.
For large values of $N$, this complexity is substantially smaller than that required to find the inverse of an $N^2 \times N^2$ matrix.

\vspace{2mm}
\noindent \textbf{Step 2: Efficient graph denoising.}
Since the optimization in \eqref{eq:step2_graph_denoising} involves $N^2$ variables (the entries in $\bbS$), using an off-the-shelf convex solver incurs a computational complexity of $\ccalO(N^7)$ \cite{boyd2004convex}.
Inspired by the Lasso regression algorithm~\cite{hastie2015statistical}, we optimize individually over each entry $S_{ij}$ in an iterative manner.
The main idea is running multiple rounds of $N^2$ efficient scalar optimizations rather than dealing with a single but demanding $N^2$-dimensional problem. To provide the details of the scheme developed to estimate $\bbS$, we need to specify the set of constraints $\ccalS$ and introduce some definitions.
Let us focus on the set of adjacency matrices $\ccalS_\ccalA := \{ \bbS | S_{ij} \geq 0,\; S_{ii} = 0 \}$ and define the vectors $\bbs := \vvec(\bbS)$, vector $\barbs := \vvec(\barbS)$, and the matrix $\bbSigma^{(t)} := \bbH^{(t+1)^\top} \oplus -\bbH^{(t+1)}$.
With these definitions in place, the minimization in \eqref{eq:step2_graph_denoising} is equivalent to solving
\begin{alignat}{2}\label{eq:efficient_step2}
    \!&\!\min_{\bbs} && \sum_{i=1}^{N^2} \left(  \lambda \bar{\omega}^{(t)}_i|s_i - \bar{s}_i| + \beta \omega^{(t)}_i s_i \right) + \gamma \|\bbSigma^{(t)}\bbs\|_2^2, \nonumber \\
    \!&\! \mathrm{\;\;s. \;to: } && \;\; \bbs \geq 0, \;\; \bbs_\ccalD = 0,
\end{alignat}
where $\bbs_\ccalD$ collects the elements in the diagonal of $\bbS$, and the vectors $\bar{\bbomega}^{(t)}$ and $\bbomega^{(t)}$ are computed according to \eqref{eq:log_weights} but with $\barbs^{(t)}$ and $\bbs^{(t)}$ in lieu of $\barbS^{(t)}$ and $\bbS^{(t)}$. The constraint $\bbs_\ccalD = 0$, implies that only the $N^2-N$ elements of $\bbs$ representing the off-diagonal entries of $\bbS$ need to be optimized. The key point to find those $N^2-N$ values is to leverage that the non-differentiable part of the cost in \eqref{eq:efficient_step2} is separable across $s_i$, postulate $N^2-N$ scalar optimization problems (coupled via the $\ell_2$ term in the cost), and address the optimization  following a projected cyclic coordinate descent scheme.

To define clearly the operation of Step 2 at each iteration $t$ of the outer loop, we need to introduce some notation. First, let us denote as $\tau$ the iteration index for the inner loop, define the sequence of variables  $\bchks^{(\tau)}$ where $\tau=0,...,\tau_{max_2}$, and initialize the sequence as $\bchks^{(0)} = \bbs^{(t)}$. Moreover, with $\ell \not\in \ccalD$ denoting an index of the off-diagonal elements of the GSO, let $\bbsigma_\ell \in \reals^{N^2}$ denote the associated $\ell$-th column of $\bbSigma^{(t)}$, $\omega_\ell\geq 0$ and $\bar{\omega}_\ell\geq 0$ the associated entries of $\bbomega^{(t)}$ and $\bar{\bbomega}^{(t)}$, and $\check{s}_\ell^{(\tau)}\in\reals$ the associated entry of $\bchks^{(\tau)}$ (note that dependence on $t$ was omitted to facilitate readability). Then, at every iteration $\tau=0,...,\tau_{max_2}-1$ of the inner loop, Alg.~\ref{A:efficient_rfi_alg} optimizes over each $\check{s}_\ell$ separately in a cyclic (successive) way.
The advantage of this approach is that the solution to the \emph{scalar} optimization over $\check{s}_\ell$ is given in closed form by the following projected soft-thresholding operation
\begin{equation}\label{eq:soft-thresholding}
    \check{s}_\ell^{(\tau+1)} = \left\{\hspace{-2mm} \begin{array}{cl}
        \left( - \bar{\lambda}_\ell + u^{(\tau)}_\ell  \right)^+ 
        &  \mathrm{if} \; \bar{s}_\ell <  - \bar{\lambda}_\ell+u^{(\tau)}_\ell,\\
        \left(\bar{\lambda}_\ell + u^{(\tau)}_\ell  \right)^+
        & \mathrm{if} \; \bar{s}_\ell >   \bar{\lambda}_\ell+u^{(\tau)}_\ell, \\
        \bar{s}_\ell & \mathrm{otherwise},
    \end{array} \right.
\end{equation}
\begin{equation}
    \mathrm{with} \;\;\; \bar{\lambda}_\ell = \frac{\lambda \bar{\omega}_\ell}{\gamma \bbsigma_\ell^\top\bbsigma_\ell}\;\; 
    \mathrm{and} \;\;\; u^{(\tau)}_\ell = \frac{-\beta \omega_\ell - \gamma\bbsigma_\ell^\top \bbr_\ell^{(\tau)} }{\gamma \bbsigma_\ell^\top\bbsigma_\ell} . \nonumber
\end{equation}
Here, $(\cdot)^+$ denotes the operation $(x)^+ = \max(0,x)$, and
\begin{equation}\label{eq:r_l}
    \bbr_\ell^{(\tau)} := \sum_{j < \ell} \bbsigma_j \check{s}_j^{(\tau+1)} + \sum_{j > \ell} \bbsigma_j \check{s}_j^{(\tau)}.
\end{equation}
Note that \eqref{eq:soft-thresholding} is a soft-thresholding operation with respect to the term $|s_i - \bar{s}_i|$.
Also, the constraints in $\ccalS_\ccalA$ are satisfied due to the projection operator $(\cdot)^+\!:=\!\max\{\cdot,\!0\}$, and because we do not optimize over the elements of the diagonal of $\bbS$.

At first sight, computing each $\check{s}_\ell$ requires roughly $N^2$ operations, so estimating the whole vector $\bbs$ would entail a computational complexity of $\ccalO(N^4)$.
However, a closer inspection of the vectors $\bbsigma_\ell$ reveals that no more than $2N$ of their entries are non-zero because $\bbsigma_\ell$ are the columns of the Kronecker sum of two $N \times N$ matrices.
We exploit this sparsity and reduce the number of operations required to compute each $s_\ell$ to approximately $2N$, rendering the final computational complexity of the graph denoising step to $\ccalO(2\tau_{max_2}N^3)$.

%%%%%%%%%%  ALGORITHM  %%%%%%%%%%
\begin{algorithm}[tb]
\footnotesize
\SetKwInput{Input}{Input}
\SetKwInOut{Output}{Output}
\Input{$\bbX$, $\bbY$, $\barbS$}
\Output{$\hbH$, $\hbS$.}
\SetAlgoLined
Initialize $\bbH^{(0)}$ and $\bbS^{(0)}$ \\
$\barbs = \vvec(\barbS)$ \\
\For{$t=0$ \KwTo $t_{max}-1$}{
    \tcp{GF-identification step}
    $\bchkH^{(0)} = \bbH^{(t)}$ \\
    \For{$\tau=0$ \KwTo $\tau_{max_1}-1$}{
        $\bchkH^{(\tau+1)} = \bchkH^{(\tau)} + \mu \nabla f_1(\bchkH^{(\tau)}|\bbS^{(t)})$ \\
    }
    $\bbH^{(t+1)} = \bchkH^{(\tau_{max_1})}$ \\
    
    %\DontPrintSemicolon \;
    \vspace{.2cm}
    \tcp{Graph denoising step}
    $[\bbsigma_1,...,\bbsigma_{N^2}] =\bbH^{(t+1)^\top} \oplus \bbH^{(t+1)}$ \\ 
    $\bchks^{(0)} = \vvec(\bbS^{(t)})$ \\
    Update $\bar{\bbomega}^{(t)}$, $\bbomega^{(t)}$ via \eqref{eq:log_weights} using $\barbs$ and $\bchks^{(0)}$ \\
    
    \For{$i=0$ \KwTo $\tau_{max_2}-1$}{
    \For{$\ell \not\in \ccalD$}{
    Obtain $\bbr_\ell^{(\tau)}$ via \eqref{eq:r_l} \\
    Obtain $\check{s}_\ell^{(\tau+1)}$ via \eqref{eq:soft-thresholding} using $\bbsigma_\ell$, $\bbr_\ell^{(\tau)}$, $\omega_\ell$, $\bar{\omega}_\ell$ \\
    }
    }
    $\bbS^{(t+1)} = \mathrm{unvec}(\bchks^{(\tau_{max_2})})$
}
$\hbH = \bbH^{(t_{max})},\; \hbS = \bbS^{(t_{max})}$.
\caption{\small Reduced-complexity robust GF identification.}
%\vspace{-0.5cm}
\label{A:efficient_rfi_alg}
\end{algorithm}
%%%%%%%%%%  END ALGORITHM  %%%%%%%%%%

%The pseudocode describing all the steps in Alg.~\ref{A:efficient_rfi_alg} is provided at the top of this page. 
The pseudocode describing the efficient implementation of Steps~1 and~2 is provided in Alg.~\ref{A:efficient_rfi_alg}.
The summary is as follows. We postulate a nested algorithm with two loops. The outer loop runs $t_{\max}$ iterations. The inner loop runs two steps: Step~1, with $\tau_{max_1}$ iterations, and Step 2, with $\tau_{\max_2}$ iterations. While the complexity for Alg.~\ref{A:rfi_alg} scaled as $\ccalO(t_{\max} N^7)$, with $t_{\max}$ being typically small, the overall computational complexity of Alg.~\ref{A:efficient_rfi_alg} is roughly $\ccalO(t_{\max}(\tau_{max_1}+\tau_{max_2})N^3)$, which is encouraging, since $2N^2$ variables are optimized and it scales with $N$ significantly better than Alg.~\ref{A:rfi_alg}. Solving Steps 1 and 2 optimally requires setting large values for $\tau_{\max_1}$ and $\tau_{\max_2}$. Nonetheless, we observe that in most tested setups the approach of setting small values for $\tau_{\max_1}$ and $\tau_{\max_2}$ (at the cost of setting a slightly higher value for $t_{\max}$) typically yields a faster convergence. Finally, implementations where the number of iterations is not fixed but selected based on some convergence criterion are also sensible alternatives. 
%the original one

We close the section noting that we developed Alg.~\ref{A:efficient_rfi_alg} for the setting described in \cref{P:problem_statement} because the notation was simpler and facilitated the discussion. Nonetheless, an analogous approach may be followed for the joint estimation of $K$ GFs (cf. Sec.~\ref{S:rfi_joint}), resulting in an algorithm with complexity per GF similar to that for Alg.~\ref{A:efficient_rfi_alg}.

%%%%%%%%%%%%%%%%%%%%%%%%%%%%%%%%%%%%%%%%%%%%%%%%%%%%%%%%%%%%%%%%%%%%%%%%%%%%%%%%%%%%%%%%%%%%%%%%%%%%%%%%%%%%%%%%%
%SECTION: NUMERICAL RESULTS
%%%%%%%%%%%%%%%%%%%%%%%%%%%%%%%%%%%%%%%%%%%%%%%%%%%%%%%%%%%%%%%%%%%%%%%%%%%%%%%%%%%%%%%%%%%%%%%%%%%%%%%%%%%%%%%%%

%%%%%%%%%%%%%%%  FIRST SET OF FIGURES   %%%%%%%%%%%%%%%
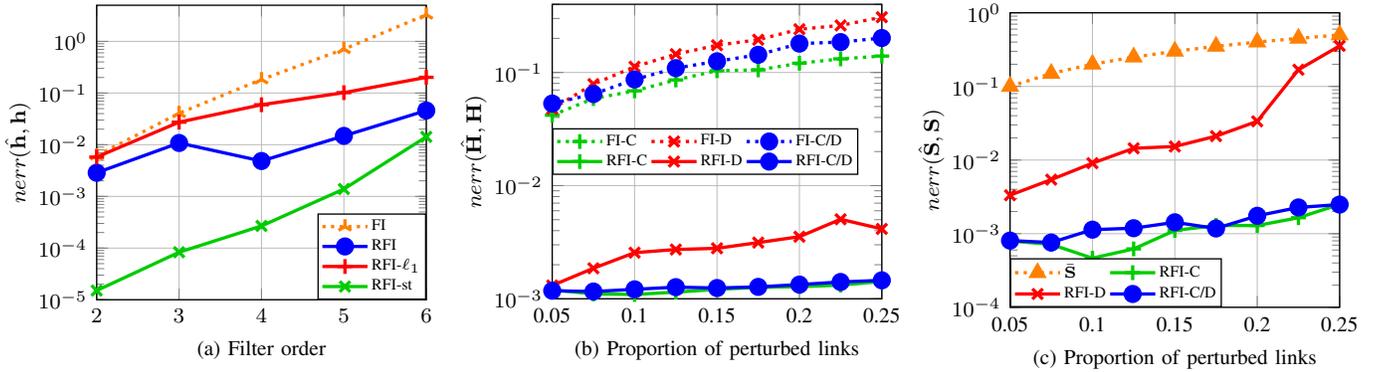
\begin{figure*}[!t]
	\centering
	\begin{subfigure}{0.32\textwidth}
		\centering
		 \begin{tikzpicture}[baseline,scale=1]

\begin{semilogyaxis}[
    %table/col sep=semicolon,
    xlabel={(a) Filter order},
    xmin={2},
    xmax={6},
    ylabel={$nerr(\hbh,\bbh)$},
    ymin={1e-5},
    ymax={5},
    ytick={1e-5,1e-4,1e-3,1e-2,1e-1,1},
    grid=major,
    legend style={
        at={(1,0)},
        anchor=south east}
    ]
    
    \pgfplotstableread{data/1-FilterOrder_decay05.csv}\timetable
    
    \addplot[orange, dotted, mark=Mercedes star] table [x={Filter Order}, y=FI] {\timetable};
    \addplot[blue, solid, mark=*] table [x={Filter Order}, y=ITER-REW-NONST] {\timetable};
    \addplot[red, solid, mark=+] table [x={Filter Order}, y=ITER-NONST] {\timetable};
    \addplot[green!80!black, solid, mark=x] table [x={Filter Order}, y=ITER-REW-ST-REAL] {\timetable};
    
    \legend{FI, RFI, RFI-$\ell_1$, RFI-st}
    
\end{semilogyaxis}
\end{tikzpicture}
	\end{subfigure}
	\begin{subfigure}{0.32\textwidth}
		\centering
		\begin{tikzpicture}[baseline,scale=1]

\begin{semilogyaxis}[
    %table/col sep=semicolon,
    xlabel={(b) Proportion of perturbed links},
    xmin={0.05},
    xmax={0.25},
    xtick={.05, .1, .15, .2, .25},
    xticklabels={0.05, 0.1, 0.15, 0.2, 0.25},
    ylabel={$nerr(\hbH,\bbH)$},
    ymin={1e-3},
    ymax={0.4},
    grid=major,
    legend style={
        at={(0,0.5)},
        anchor=west},
    legend columns=3,
    ]
    
    \pgfplotstableread{data/pert_type_rf_errH.csv}\errRFtable
    \pgfplotstableread{data/pert_type_rfi_errH.csv}\errRFItable
    
    \addplot[green!80!black, dotted, mark=+] table [x=perts, y=RF-creat] {\errRFtable};
    \addplot[red, dotted, mark=x] table [x=perts, y=RF-dest] {\errRFtable};
    \addplot[blue, dotted, mark=*] table [x=perts, y=RF-creat-dest] {\errRFtable};
    \addplot[green!80!black, mark=+] table [x=perts, y=RFI-creat] {\errRFItable};
    \addplot[red, mark=x] table [x=perts, y=RFI-dest] {\errRFItable};
    \addplot[blue, mark=*] table [x=perts, y=RFI-creat-dest] {\errRFItable};
    
    % \legend{FI-c, RFI-c, FI-d, RFI-d, FI-c/d, RFI-c/d}
    \legend{FI-C, FI-D, FI-C/D, RFI-C, RFI-D, RFI-C/D}
    
\end{semilogyaxis}
\end{tikzpicture}
	\end{subfigure}
	\begin{subfigure}{0.32\textwidth}
		\centering
		\begin{tikzpicture}[baseline,scale=1]

\begin{semilogyaxis}[
    %table/col sep=semicolon,
    xlabel={(c) Proportion of perturbed links},
    xmin={0.05},
    xmax={0.25},
    xtick={.05, .1, .15, .2, .25},
    xticklabels={0.05, 0.1, 0.15, 0.2, 0.25},
    ylabel={$nerr(\hbS,\bbS)$},
    ymin={1e-4},
    ymax={1},
    grid=major,
    legend style={
        at={(0,0)},
        anchor=south west},
    legend columns=2,
    ]
    
    \pgfplotstableread{data/pert_type_rf_errS.csv}\errRFtable
    \pgfplotstableread{data/pert_type_rfi_errS.csv}\errRFItable
    
    \addplot[orange, dotted, mark=triangle*] table [x=perts, y=RF-creat] {\errRFtable};
    % \addplot[red, dotted, mark=x] table [x=perts, y=RF-dest] {\errRFtable};
    % \addplot[blue, dotted, mark=*] table [x=perts, y=RF-creat-dest] {\errRFtable};
    \addplot[green!80!black, mark=+] table [x=perts, y=RFI-creat] {\errRFItable};
    \addplot[red, mark=x] table [x=perts, y=RFI-dest] {\errRFItable};
    \addplot[blue, mark=*] table [x=perts, y=RFI-creat-dest] {\errRFItable};
    
    \legend{$\barbS$, RFI-C, RFI-D, RFI-C/D}
    % \legend{FI-C, FI-D, FI-C/D, RFI-C, RFI-D, RFI-C/D}
    
\end{semilogyaxis}
\end{tikzpicture}
	\end{subfigure}
		\vspace{-0.3cm}
	\caption{Assessing the performance of the robust GF identification algorithm and the impact of perturbations in the topology.
	(a) shows the error of estimating $\hbh$ as the order of the GF increases; (b) and (c) respectively show the error of estimating $\hbH$ and $\hbS$ using a robust or a non-robust approach for several types of perturbations.}	\vspace{-0.5cm} \label{F:exps1}
\end{figure*}
%%%%%%%%%%%%%%%%%%%%%%%%%%%%%%%%%%%%%%%%%%%

%%%%%%%%%%%%%%%   NOW WE START WITH THE TEXT  %%%%%%%%%%%%%%%
\section{Numerical results}\label{S:experiments}
This section discusses several numerical experiments to gain insights and assess the performance of the robust GF identification algorithms.
Unless specified otherwise, for a variable of interest $\bbTheta$, we report its normalized estimation error defined as
\begin{equation}\label{eq:rel_err}
    nerr(\hbTheta, \bbTheta) := {\| \hbTheta - \bbTheta \|_F^2}/{\|\bbTheta\|_F^2},    
\end{equation}
where $\hbTheta$ and $\bbTheta$ denote the estimated and the true value, respectively.
The code implementing our algorithms and the experiments presented next is available on GitHub\footnote{\url{https://github.com/reysam93/graph_denoising}}.
%The code implementing the algorithms proposed in the paper, used to run the experiments in this section, and needed to run additional test-cases that, due to space limitations, are not presented here is available on GitHub\footnote{\url{https://github.com/reysam93/graph_denoising}}.
The interested reader is referred there for additional details and tests.

\subsection{Synthetic experiments}
We start by evaluating our algorithms with synthetic data, which is key to gain intuition.
Unless otherwise stated, graphs are sampled from an Erd\H{o}s Rényi (ER) random graph model with a link probability of $p=0.2$ and $N=20$ nodes; $\barbS$ is obtained by randomly creating and destroying 10\% of the links in $\bbS$; $M=50$ signals $\bbX$ and $\bbY$ are generated according to \eqref{eq:observation_model},  with the columns of $\bbX$ being drawn from a multivariate Gaussian distribution $\ccalN(\mathbf{0},\bbI)$, so the signals $\bbY$ are stationary on $\bbS$; signals in $\bbY$ are corrupted with white Gaussian noise with a normalized power of $\eta_\bbW=0.05$; and the reported error corresponds to the median of $nerr$ across 64 realizations of graphs and graph signals.

\vspace{2mm}
\noindent\textbf{Test case 1.} The first experiment evaluates the influence of perturbations as the order of the GF $R$ increases.
The number of observed pairs of signals considered is $M=100$ and 10\% of the edges in $\bbS$ are perturbed.
Results are reported in Fig.~2(a), where the x-axis shows $R$ and the y-axis $nerr(\hbh,\bbh)$.
The algorithms considered are: (i) the GF identification algorithm that ignores perturbations [see \eqref{eq:solving_fi}], denoted as ``FI''; (ii) the robust GF identification algorithm from Alg.~\ref{A:rfi_alg} (``RFI''); (iii) a variation of ``RFI'' where the reweighted $\ell_1$ norm is replaced by the standard $\ell_1$ norm (``RFI-$\ell_1$''); and (iv) the robust GF identification algorithm accounting for the stationarity of $\bbY$ (``RFI-ST'').
First, we observe that the error of the ``FI'' algorithm, while small for low values of $R$, increases rapidly as $R$ grows. 
This is aligned with the discussion of high-order polynomials in Sec.~\ref{S:perturbed_graph_filters} and illustrates the merits of the robust algorithms.
Moreover, ``RFI-ST'' presents the best performance illustrating the importance of exploiting additional structure when it is available.
Finally, comparing the error of ``RFI'' and ``RFI-$\ell_1$'' showcases the benefits of replacing the $\ell_1$ norm with its reweighted version.

\vspace{2mm}
\noindent\textbf{Test case 2.}
The next experiment tests the influence of different types of perturbations in the robust and non-robust GF identification algorithms.
Figs.~2(b) and~2(c) illustrate the error of the estimated GF $\hbH$ and the denoised GSO $\hbS$ 
as the ratio of perturbed links in $\barbS$ increases.
Graphs are sampled from the small world~\cite{newman1999renormalization} random graph model and $\barbS$ is obtained by creating new links, destroying existing links, or simultaneously creating and destroying links, which are respectively denoted as ``C'', ``D'', and ``C/D'' in the legend.
Since the non-robust ``FI'' algorithm does not perform graph denoising we show the error $nerr(\barbS, \bbS)$, denoted as ``$\barbS$'' in Fig.~2(c).
Furthermore, because the number of perturbed links is fixed, the error of $\barbS$ is the same for the considered perturbations and it is only plotted once.
From the figures, we observe that destroying links is the most harmful perturbation, especially when the focus is on $\hbS$.
This may be explained because destroying links is prone to produce non-connected graphs.
Nevertheless, the results show the resilience of the ``RFI'' algorithm, which provides low-error estimates $\hbH$ and $\hbS$ even when more than 20\% of the links are perturbed.

%%%%%%%%%%%%%%%   MORE FIGURES   %%%%%%%%%%%%%%%
\begin{figure*}[!t]
	\centering
	\begin{subfigure}[t]{0.32\textwidth}
		\centering
		 \begin{tikzpicture}[baseline,scale=1]
\begin{semilogyaxis}[
    xlabel={(a) Proportion of perturbed links},
    xmin={0},
    xmax={0.3},
    ylabel={$nerr(\hbH,\bbH)$},
    ymin={1e-5},
    ymax={1},
    ytick={1e-5,1e-4,1e-3,1e-2,1e-1,1},
    grid=major,
    legend style={
        at={(1,0)},
        anchor=south east},
    legend columns=2]
    
    \pgfplotstableread{data/4-ComparisonTLS-LLS.csv}\timetable
    
    \addplot[blue, solid, mark=*] table [x={Adjacency perturbation}, y=RFI-iter-H] {\timetable};
    \addplot[blue, dotted, mark=*] table [x={Adjacency perturbation}, y=RFI-iter-SEM] {\timetable};
    \addplot[red, solid, mark=x] table [x={Adjacency perturbation}, y=TLS-SEM-H] {\timetable};
    \addplot[red, dotted, mark=x] table [x={Adjacency perturbation}, y=TLS-SEM-SEM] {\timetable};
    \addplot[green!80!black, solid, mark=+] table [x={Adjacency perturbation}, y={LLS-SCP (Natali)-H}] {\timetable};
    \addplot[green!80!black, dotted, mark=+] table [x={Adjacency perturbation}, y={LLS-SCP (Natali)-SEM}] {\timetable};
    
    \legend{{RFI, H}, {RFI, SEM}, {TLS-SEM, H}, {TLS-SEM, SEM}, {SCP, H}, {SCP, SEM}}

\end{semilogyaxis}
\end{tikzpicture}
	\end{subfigure}
	\begin{subfigure}[t]{0.32\textwidth}
		\centering
		\begin{tikzpicture}[baseline,scale=1]

\begin{semilogyaxis}[
    table/col sep=comma,
    xlabel={(b) Number of nodes},
    xmin={20},
    xmax={100},
    ylabel={time (s)},
    ymax={1e5},
    ytick={1,1e1,1e2,1e3,1e4,1e5},
    grid=major,
    legend style={
        at={(0,1)},
        anchor=north west,
        },
    ]
    
    \pgfplotstableread{data/eff_alg_mean_times.csv}\timetable
    
    \addplot[blue, solid, mark=*] table [x=x, y=alg1] {\timetable};
    \addplot[red, loosely dotted, mark=x] table [x=x, y=alg2] {\timetable};
    \addplot[orange, dotted, mark=x] table [x=x, y=alg3] {\timetable};
    \addplot[green!80!black, densely dotted, mark=x] table [x=x, y=alg4] {\timetable};
    
    \legend{Stand-5, Eff-5-10, Eff-5-25, Eff-5-50}
    
\end{semilogyaxis}
\end{tikzpicture}
	\end{subfigure}
	\begin{subfigure}[t]{0.32\textwidth}
		\centering
		\begin{tikzpicture}[baseline,scale=1]

\begin{axis}[
    table/col sep=comma,
    xlabel={(c) Number of nodes},
    xmin={20},
    xmax={100},
    ylabel={$nerr(\hbH,\bbH)$},
    ymin={0},
    ymax={0.4},
    grid=major,
    legend style={
        at={(0,1)},
        anchor=north west}
    ]
    
    \pgfplotstableread{data/eff_alg_mean_err_H.csv}\timetable
    
    \addplot[blue, solid, mark=*] table [x=x, y=alg1] {\timetable};
    \addplot[red, loosely dotted, mark=x] table [x=x, y=alg2] {\timetable};
    \addplot[orange, dotted, mark=x] table [x=x, y=alg3] {\timetable};
    \addplot[green!80!black, densely dotted, mark=x] table [x=x, y=alg4] {\timetable};
    
    \legend{Stand-5, Eff-5-10, Eff-5-25, Eff-5-50}
    
\end{axis}
\end{tikzpicture}
	\end{subfigure}
		\vspace{-0.6cm}
	\caption{Comparing the performance of several robust GF identification algorithms.
	(a) shows the error of $\hbH$ when estimated with the proposed algorithm and with other baselines as the ratio of perturbed links increases.
	Different graph-signal models are considered.
	(b) and (c) respectively show the running time and error of $\hbH$ using Alg.~\ref{A:rfi_alg} and Alg.~\ref{A:efficient_rfi_alg} as the number of nodes increases.
	Different values for the maximum number of iterations of the inner loops are considered.}	\vspace{-0.5cm} \label{F:exps2}
\end{figure*}
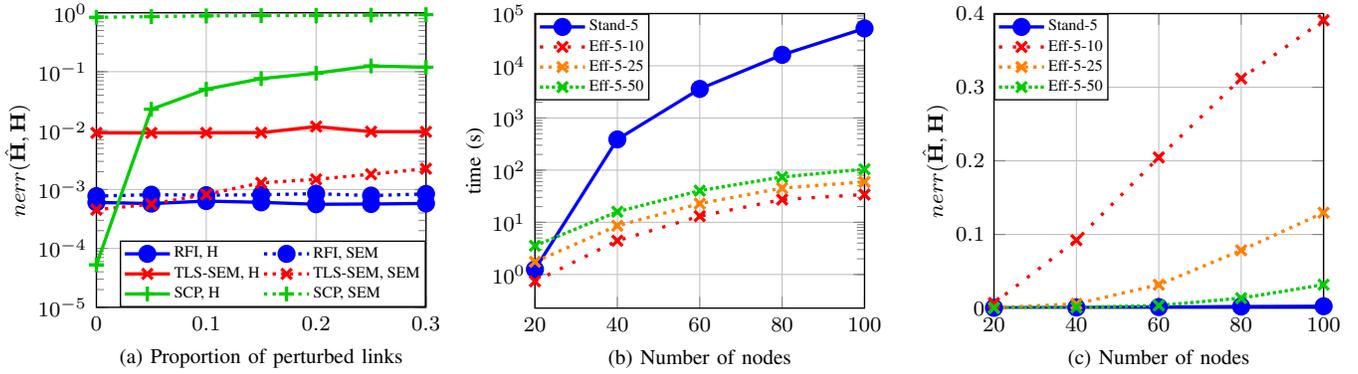
%%%%%%%%%%%%%%%%%%%%%%%%%%%%%%%%%%%%%%%%%%%

\vspace{2mm}
\noindent\textbf{Test case 3.}
Next, we compare the performance of our algorithms with other robust alternatives.
Fig.~3(a) reports, for each algorithm, $nerr(\hbH,\bbH)$ as the ratio of perturbed links increases.
The baselines considered are the TLS-SEM algorithm from \cite{ceci2020_semtls}, and LLS-SCP from \cite{natali2020topology}.
We note that the TLS-SEM algorithm is tailored to graph signals following a SEM of the form 
\begin{equation}\label{eq:sem_model}
    \bbY=\bbA\bbY+\bbX = (\bbI-\bbA)^{-1}\bbX,
\end{equation}
where the observations at the $i$-th node are represented by the values of the neighbors of $i$ and an exogenous input.
As a result, the TLS-SEM algorithm may not be well suited to deal with signals generated according to the more general model in \eqref{eq:observation_model}.
Taking this into account, to offer a more favorable comparison we consider two types of graph signals: (i) signals generated according to \eqref{eq:sem_model}, denoted as ``SEM''; and (ii) signals generated according to \eqref{eq:observation_model}, denoted as ``H''.
It is worth noting that the ``SEM'' can be considered as a particular case of the model ``H'' when the GF $\bbH_{SEM} = (\bbI-\bbA)^{-1}$ is employed.
%Regarding \cite{natali2020topology}, since the perturbations considered create and destroy links, knowing the true support of $\bbS$ is tantamount to knowing $\bbS$. Therefore, we employ a version of LLS-SCP that does not exploit the knowledge of the support of $\bbS$

Looking at the results in Fig.~3(a) we observe the following.
When the ``SEM'' model is considered, TLS-SEM (denoted as ``TLS-SEM'') obtains the best performance when the perturbation probability is small, and then, the performance of ``TLS-SEM'' and that of the ``RFI'' algorithm become comparable.
This illustrates that our algorithm is especially suitable to deal with a large number of perturbed links.
On the other hand, when the ``H'' model is considered, we observe that the ``RFI'' algorithm consistently outperforms the baselines in the presence of perturbations.
The good performance of the ``RFI'' algorithm on both signal models highlights the flexibility of the proposed formulation since it considers more lenient assumptions than the other alternatives.

\vspace{2mm}
\noindent\textbf{Test case 4.}
Now, we compare the performance of the standard and the efficient implementation of the robust identification algorithm, as described in Algs.~\ref{A:rfi_alg} and \ref{A:efficient_rfi_alg}.
The results are shown in Figs.~3(b) and 3(c), where the figures depict the running time measured in seconds and $nerr(\hbH,\bbH)$ as $N$ increases.
The legend identifies first the algorithm employed, then the number of iterations of the outer loop ($t_{max}$), and finally the iterations of the inner loops (with $\tau_{{max_1}}=\tau_{{max_2}}$).
As expected, Fig.~3(b) shows that Alg.~\ref{A:efficient_rfi_alg} is remarkably faster than Alg.~\ref{A:rfi_alg} even with medium-sized graphs, achieving a running time $10^3$ times smaller when $N=100$.
On the other hand, in Fig.~3(c) we observe that ``Eff-5-50'' has an error that is close to the standard implementation (``Stand-5'') even though it is considerably faster.
Furthermore, the trade-off between speed and estimation accuracy is also evident.
``Eff-5-10'' is the fastest implementation but the quality of its estimated GF may not be enough for graphs with more than 40 nodes.

\vspace{2mm}
\noindent\textbf{Test case 5.} 
The last experiment with synthetic data studies the benefits of the joint GF estimation.
All the GFs are polynomials of the same $\bbS$, and for each $\bbH_k$ we consider $M_k = 15$ noisy observations with $\eta_\bbw = 0.01$.
Fig.~4 shows the results, with the y-axis being the normalized error averaged across the $K$ graphs, i.e., $\frac{1}{K}\sum_{k=1}^K nerr(\hbH_k, \bbH_k)$, and the x-axis representing $K$.
We compare the performance of estimating the GFs jointly (marked as ``J'' in the legend) or separately for the three algorithms (``RFI-$\ell_1$'', ``RFI'', and ``RFI-st'') described in Test case~1.
Note that ``RFI-J'' corresponds to the formulation in \eqref{eq:joint_rfi_noncvx_rew}.
The first thing we observe from the results in Fig.~4 is that the error decreases as $K$ increases when a joint algorithm is employed.
This is aligned with the discussion in Sec.~\ref{S:rfi_joint} and illustrates the benefit of exploiting the common structure.
In addition, algorithms accounting for the stationary of $\bbY$ outperform the non-stationary alternatives even though we only have $M=15$ signals to estimate the covariance $\hbC_\bby$.

%%%%%%%%%%%%%%%   MORE FIGURES   %%%%%%%%%%%%%%%
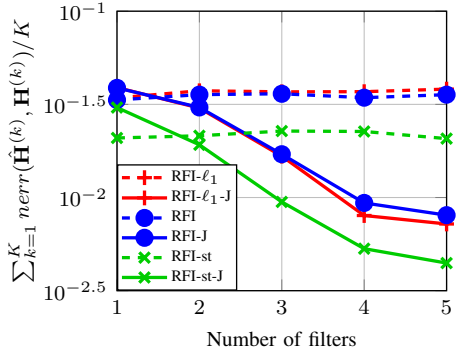
\begin{figure}[]
	\centering
	\centering
		\begin{tikzpicture}[baseline,scale=1]

\begin{semilogyaxis}[
    %table/col sep=semicolon,
    width=0.32\textwidth,
    height=5.3cm,
    xlabel={Number of filters},
    xmin={1},
    xmax={5},
    ylabel={$\sum_{k=1}^K nerr(\hbH^{(k)},\bbH^{(k)})/K$},
    ymin={10^-2.5},
    ymax={1e-1},
    ytick={10^-2.5, 1e-2, 10^-1.5, .1},
    grid=major,
    legend style={
        at={(0,0)},
        anchor=south west},
    ]
    
    \pgfplotstableread{data/5-Joint-RFI.csv}\timetable
    
    \addplot[red, dashed, mark=+] table [x={Number of filters}, y=RFI-iter] {\timetable};
    \addplot[red, mark=+] table [x={Number of filters}, y=RFI-iter-J] {\timetable};
    \addplot[blue, dashed, mark=*] table [x={Number of filters}, y=RFI-iter-rew] {\timetable};
    \addplot[blue, mark=*] table [x={Number of filters}, y=RFI-iter-rew-J] {\timetable};
    \addplot[green!80!black, dashed, mark=x] table [x={Number of filters}, y=RFI-iter-Cyest] {\timetable};
    \addplot[green!80!black, mark=x] table [x={Number of filters}, y=RFI-iter-Cyest-J] {\timetable};

    % \addplot[red, solid, mark=+] table [x={Number of filters}, y=ITER-NONST] {\timetable};
    % \addplot[blue, solid, mark=*] table [x={Number of filters}, y=ITER-REW-NONST] {\timetable};
    % \addplot[green!80!black, solid, mark=x] table [x={Number of filters}, y=ITER-REW-ST-REAL] {\timetable};
    
    \legend{RFI-$\ell_1$, RFI-$\ell_1$-J, RFI, RFI-J, RFI-st, RFI-st-J}
    
\end{semilogyaxis}
\end{tikzpicture}
			\vspace{-0.15cm}
	\caption{Error performance when estimating $K$ GFs using the separate and joint approach for different values of $K$. %The reported error is the median across 64 (graph and perturbation) realizations.
	}\label{F:exp_jointfi}
	%\vspace{-0.15cm}
\end{figure}

\subsection{Real-world datasets}
To close the numerical evaluation, we test our robust GF identification algorithms over two real-world datasets.

%%%%%%%%%%%%%%%   MORE FIGURES   %%%%%%%%%%%%%%%
\begin{figure}[]
	\centering
	\centering
	\begin{tikzpicture}[baseline,scale=1]

\begin{semilogyaxis}[
    %table/col sep=semicolon,
    width=0.32\textwidth,
    height=5.3cm,
    xlabel={Time horizon used for prediction},
    xmin={1},
    xmax={5},
    xtick={1, 2, 3, 4, 5},
    xticklabels={1, 2, 3, 4, 5},
    ylabel={$\sum_{\kappa=1}^M nerr(\hby_\kappa,\bby_\kappa) / M$},
    ymin={1e-2},
    ymax={0.1},
    ytick={1e-2, 10^-1.5, .1},
    grid=major,
    legend style={
        at={(1,0)},
        anchor=south east},
    legend columns=2,
    ]
    
    \pgfplotstableread{data/airQuality-NSteps.csv}\errNSteps
    
    \addplot[green!80!black, dotted, mark=+] table [x=N-Steps, y=LS-Perfect-(LB)] {\errNSteps};
    \addplot[blue, mark=*] table [x=N-Steps, y=Least-Squares] {\errNSteps};
    \addplot[red, mark=+] table [x=N-Steps, y=Copy-Prev-Day] {\errNSteps};
    \addplot[orange, mark=+] table [x=N-Steps, y=TLS-SEM] {\errNSteps};
    \addplot[purple, mark=*] table [x=N-Steps, y=Least-Squares-GF] {\errNSteps};
    \addplot[cyan, mark=x] table [x=N-Steps, y=RGFI] {\errNSteps};
    \addplot[black, mark=x] table [x=N-Steps, y=VAR-RGFI] {\errNSteps};
    
    % \legend{FI-c, RFI-c, FI-d, RFI-d, FI-c/d, RFI-c/d}
    \legend{{LS-Eval (LB)},LS,Copy-Prev-Day,TLS-SEM,LS-GF,RFI,AR(3)-RFI}
    
\end{semilogyaxis}
\end{tikzpicture}
	\vspace{-0.15cm}
	\caption{Performance of the algorithms predicting ozone levels in the AirData station network, as the time horizon of the prediction increases.} \label{F:exp_airQual}
	% 	\vspace{-0.5cm}
\end{figure}
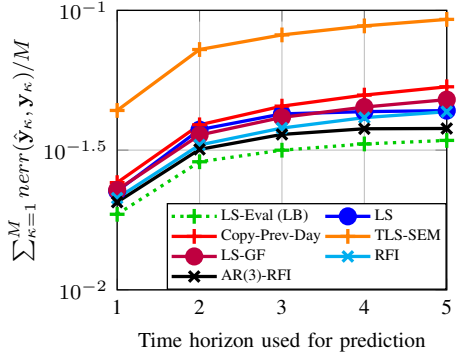
%%%%%%%%%%%%%%%%%%%%%%%%%%%%%%%%%%%%%%%%%%%

\begin{table}[!t]
	\centering
	%\setlength{\tabcolsep}{2pt}
% \begin{tabular}{c||c|c||c|c|}
% \multirow{2}{*}{Models} & \multicolumn{2}{c}{1-Step} & \multicolumn{2}{c}{3-Step} \\
% & TTS = 0.25 & TTS = 0.5 & TTS = 0.25 & TTS = 0.5 \\ \hline
% \!\!LS & 6.9e-3 & 3.1e-3 & 2.1e-2 & 9.1e-3 \\
% \!\!LS-GF & 3.3e-3 & 3.3e-3 & 8.4e-3 & 8.5e-3 \\
% \!\!TLS-SEM & 4.0e+1 & 3.7e-2 & 6.8e-1 & 5.5e-2 \\
% \!\!RFI & 3.4e-3 & 3.1e-3 & 8.5e-3 & 7.5e-3 \\
% \!\!VAR-RFI & \textbf{3.2e-3} & \textbf{2.8e-3} & \textbf{7.8e-3} & \textbf{6.9e-3} \\
% \end{tabular}

\begin{tabular}{c||c|c||c|c|}
\multirow{2}{*}{Models} & \multicolumn{2}{c}{1-Step} & \multicolumn{2}{c}{3-Step} \\
& TTS=0.25 & TTS = 0.5 & TTS=0.25 & TTS = 0.5 \\ \hline
\!\!LS & $6.9 \cdot 10^{-3}$ & $3.1 \cdot 10^{-3}$ & $2.1 \cdot 10^{-2}$ & $9.1 \cdot 10^{-3}$ \\
\!\!LS-GF & $3.3 \cdot 10^{-3}$ & $3.3 \cdot 10^{-3}$ & $8.4 \cdot 10^{-3}$ & $8.5 \cdot 10^{-3}$ \\
\!\!TLS-SEM & $4.0 \cdot 10^{1}$ & $3.7 \cdot 10^{-2}$ & $6.8 \cdot 10^{-1}$ & $5.5 \cdot 10^{-2}$ \\
\!\!RFI & $3.4 \cdot 10^{-3}$ & $3.1 \cdot 10^{-3}$ & $8.5 \cdot 10^{-3}$ & $7.5 \cdot 10^{-3}$ \\
\!\!AR(3)-RFI & \textbf{$3.2 \cdot 10^{-3}$} & \textbf{$2.8 \cdot 10^{-3}$} & \textbf{$7.8 \cdot 10^{-3}$} & \textbf{$6.9 \cdot 10^{-3}$} \\
\end{tabular}
		%\vspace{-0.3}
	\caption{Performance of the algorithms in predicting the temperature for 2 prediction horizons (1 and 3) and 2 values (25\% and 50\%) of train-test split (TTS). The metrics shown are the average of the normalized error at each timestep $\frac{1}{M} \sum_{\kappa=1}^M nerr(\hby_{\kappa},\bby_{\kappa})$ for all samples.} \label{T:tempData}
    \vspace{0.1cm}
\end{table}

\vspace{2mm}
\noindent\textbf{Weather station network.} This test case evaluates the ability of our algorithms to predict the temperature measured by a network of stations using the data from previous days. The data comes from the ``Global Summary of the Day'' dataset of the National Centers for Environmental Information\footnote{\url{https://www.ncei.noaa.gov/data/global-summary-of-the-day/archive/}} and we used daily temperature measurements from $N=17$ stations in California during 2017 \& 2018. Specifically, with $\bby_\kappa \in \reals^N$ collecting the measurements of the 17 stations at day $\kappa$, we consider an AR model without exogenous inputs, so that $\bby_\kappa \approx \sum_{k=1}^K \bbH_k \bby_{\kappa-k}$. The data samples were divided into two subsets, the first one (training) was used to obtain the GFs $\bbH_k$ and the second one (evaluation) was used to assess the performance and the generalization power of the GFs obtained. Also, the data is normalized so that the signal at each station for all time samples has unitary norm.

The underlying $\ccalG$ was constructed as the unweighted 5-nearest neighbors graph, using the geographical distance between stations. Since temperature relations across stations are likely to be due to a range of factors (including, e.g., altitude), the considered adjacency (based only on geographical positions) may be imperfect, rendering our robust algorithms better suited for this task.

The estimation performance of the different algorithms is shown in Table~\ref{T:tempData}. Since in this case the ground-truth GF is not known, we use the signal denoising error $nerr(\bby_{\kappa},\hby_{\kappa})$ to assess the quality of the schemes. In this specific experiment, the error is measured over all samples (both training and test subsets), to see a clear downward trend when increasing the number of training samples, or equivalently, the train-test split (TTS) value.
The algorithms evaluated are ``LS'', ``LS-GF'' (which postulates a GF with coefficients $\hbh = \argmin_\bbh \| \bbY - \sum_\ell h_\ell \bbS^\ell \bbX \|_F^2$), ``TLS-SEM'', ``RFI'' (which assumes an AR(1) process) and ``AR(3)-RFI''. Two values of TTS (0.25 and 0.50) and two prediction horizons (1 and 3) are considered. 
The main observation is that ``AR(3)-RFI'' yields the best performance in all settings. Additionally, the results for TTS=0.25 demonstrate the benefits of considering the underlying graph in the low-sample regime, since even ``LS-GF'', which relies on the  imperfect $\bar{\bbS}$, outperforms ``LS''. On the other hand, ``LS-GF'' does not seem to improve its prediction as TTS increases, while our two algorithms yield a lower prediction error.
%RGFI is beaten by LS-GF for low values of TTS, but  

\noindent\textbf{Air quality station network.}
We consider an experimental setup (AR model, graph creation method...) similar to that for the weather station data but, in this case, we use 2018 \& 2019 data from the United States Environmental Protection Agency\footnote{\url{https://www.epa.gov/outdoor-air-quality-data}} to predict the ozone levels in a network of 17 outdoor stations in California.
The stations chosen were those with at least 330 measurements each year for a selection of pollutants, and missing data was filled via first-order interpolation.

The goal here is to analyze how the prediction horizon affects the prediction error.
%As in the previous experiment, we used part of the data to obtain the GFs and then evaluated on unseen data.
The value of TTS chosen was 0.5, i.e. evaluation data represented 50\% of the samples.
Fig.~\ref{F:exp_airQual} shows the performance of the algorithms when predicting ozone levels.
As a baseline, ``LS-Eval-(LB)'' shows the error measured on the evaluation data when obtaining the GF also using evaluation data, therefore representing a lower bound for the LS error using AR models of order 1. Also, ``Copy-Prev-Day'' represents the error obtained by the ``identity GF'', which copies the previous day's measurement.
As in the previous example, the best performing algorithm is ``AR(3)-RFI'', whose performance is close to the baseline, followed by ``RFI''.

%%%%%%%%%%%%%%%%%%%%%%%%%%%%%%%%%%%%%%%%%%%%%%%%%%%%%%%%%%%%%%%%%%%%%%%%%%%%%%%%%%%%%%%%%%%%%%%%%%%%%%%%%%%%%%%%%
%SECTION: CONCLUDING REMARKS
%%%%%%%%%%%%%%%%%%%%%%%%%%%%%%%%%%%%%%%%%%%%%%%%%%%%%%%%%%%%%%%%%%%%%%%%%%%%%%%%%%%%%%%%%%%%%%%%%%%%%%%%%%%%%%%%%
%
\section{Concluding remarks}\label{S:conclusion}
This paper put forth a framework dealing with estimation problems in GSP where the information about (the links of) the supporting graph is uncertain. 
Specifically, we addressed the problem of estimating a GF (i.e., a polynomial of the GSO) from input and output graph signals under the key assumption that only a perturbed version of the true GSO was available. 
In contrast to the majority of existing approaches that operate on the spectral domain, we recast the true graph as an additional estimation variable and formulated an optimization problem that \emph{jointly} estimated the GF and the true (unknown) GSO. We focused first on the case where only one GF needed to be estimated and, then, shifted to (multi-feature and AR graph signal) setups where multiple GFs have to be jointly identified. 
The formulated optimizations operated completely in the vertex domain and bypassed the problem of computing high-order polynomials, avoiding the challenges of dealing with the influence of perturbations in the graph spectrum as well as the numerical instability and error propagation associated with  high-order matrix polynomials. While non-convex, upon blending techniques from alternating optimization and MM, the proposed algorithm was shown to be capable to find a stationary point in polynomial time. This algorithm was later modified so that the scaling of the computational complexity with respect to the number of nodes in the graph is reduced. Future work includes delving into the robust estimation of ARMA time-varying graph signals, consideration of additional graph perturbation models, and application of our robust estimation framework to other GSP problems, to name a few.

%%%%%%%%%%%%%%%%%%%%%%%%%%%%%%%%%%%%%%%%%%%%%%%%%%%%%%%%%%%%%%%%%%%%%%%%%%%%%%%%%%%%%%%%%%%%%%%%%%%%%%%%%%%%%%%%%
%%%%%%%%%%%%%%%%%%%%%%%%%%%%%%%%%%%%%%%%%%%%%%%%%%%%%%%%%%%%%%%%%%%%%%%%%%%%%%%%%%%%%%%%%%%%%%%%%%%%%%%%%%%%%%%%%
%APPENDICES AND REFERENCES
%%%%%%%%%%%%%%%%%%%%%%%%%%%%%%%%%%%%%%%%%%%%%%%%%%%%%%%%%%%%%%%%%%%%%%%%%%%%%%%%%%%%%%%%%%%%%%%%%%%%%%%%%%%%%%%%%
%%%%%%%%%%%%%%%%%%%%%%%%%%%%%%%%%%%%%%%%%%%%%%%%%%%%%%%%%%%%%%%%%%%%%%%%%%%%%%%%%%%%%%%%%%%%%%%%%%%%%%%%%%%%%%%%%
%\appendices
\section*{Appendix A: Proof of Th.~\ref{thm1}}\label{A:proof_thm1}
The proof relies on the results presented in \cite[Th. 1b]{hong2015unified}, so it suffices to show that our formulation and algorithm fulfill the required conditions in \cite{hong2015unified}. To that end, recall that $f(\bbz)$ is the objective function in \eqref{eq:rfi_nonconvex_rew}, and let $\bbz_1:=\vvec(\bbH)$ and $\bbz_2:=\vvec(\bbS)$ denote the $B=2$ blocks of variables considered in our algorithm.
Moreover, at each step, the function $f(\bbz)$ is approximated by $u_1(\bbz_1)$ and $u_2(\bbz_2)$, corresponding to the objective functions in \eqref{eq:step1_filterid} and \eqref{eq:step2_graph_denoising}.
Then, to ensure the convergence of our iterative algorithm the following conditions are required.

% \noindent\textit{(\textbf{C1}) Each of the approximation functions $u_b(\bbz_b)$ must be a global upper bound of $f(\bbz)$ and the first order behavior of $u_b(\bbz_b)$ and $f(\bbz)$ must be the same.}
\noindent\textit{(\textbf{C1}) Each function $u_b(\bbz_b)$ must be a global upper bound of $f(\bbz)$ and the first-order behavior of $u_b(\bbz_b)$ and $f(\bbz)$ must be the same.}

\noindent\textit{(\textbf{C2}) $f(\bbz)$ must be regular (cf. \cite{hong2015unified}) at every point in $\ccalZ^*$.} 

\noindent \textit{(\textbf{C3}) The level set $\ccalZ^{(0)} = \{\bbz \; | \; f(\bbz) \leq f(\bbz^{(0)}) \}$ is compact.}

\noindent\textit{(\textbf{C4}) At least one of the problems in \eqref{eq:step1_filterid} and \eqref{eq:step2_graph_denoising} must have a unique solution.}

\noindent
Next, we address each of the four conditions separately, proving that our approach satisfies all of them.

Condition (\textbf{\textit{C1}}) requires the surrogate functions $u_b(\bbz_b)$ to be global upper bounds of $f(\bbz)$.
For the first block ($b=1$), it is easy to see that $u_1(\bbz_1)=f(\bbz)$ when the block $\bbz_2$ remains constant, so it satisfies the requirements.
Regarding $u_2(\bbz_2)$, we approximate $f(\bbz)$ with the first-order Taylor series of the logarithmic penalty.
Because the $\log$ is a concave differentiable function, it follows that its Taylor series of order one constitutes a global upper bound. 
Moreover, because $u_2(\bbz_2)$ is a first-order Taylor series approximation of $f(\bbz)$, it also follows that the first-order behavior of $f(\bbz)$ and $u_2(\bbz_2)$ is the same.
Therefore, $u_2$ also satisfies the requirement, and hence, (\textbf{\textit{C1}}) is fulfilled.

To prove (\textit{\textbf{C2}}), according to \cite{hong2015unified}, a function $f(\bbz)$ is regular if its non-smooth components are separable across the different blocks of variables.
To show this, we decompose $f$ as $f = g_A+g_B$, with functions $g_A$ and $g_B$ being defined as

\begin{align}
    &g_A(\bbH, \bbS) = \|\bbY - \bbH \bbX\|_F^2 + \gamma \|\bbH \bbS - \bbS \bbH\|_F^2, \nonumber \\
    &g_B(\bbS) = \lambda \sum_{i,j=1}^N\log(|S_{ij}\!-\!\bar{S}_{ij}|+\delta_2) + \beta\sum_{i,j=1}^N\log(|S_{ij}|+\delta_1). \nonumber
\end{align}
Since $g_A$ is a smooth function and the non-smooth function $g_B$ only depends on the variables on the second block, $\bbz_2=\vvec(\bbS)$, it follows that $f(\bbz)$ is a regular function for all feasible points.

Next, we show that the level set $\ccalZ^{(0)} = \{\bbz \; | \; f(\bbz) \leq f(\bbz^{(0)}) \}$ is compact as required by (\textit{\textbf{C3}}).
We start by noting that the entries of $\bbS$ are continuous subsets of $\reals$, (e.g., $S_{ij}\in\reals_+$ when $\bbS=\bbA$), and that $\bbH\in\reals^{N\times N}$, so $f(\bbz)$ is continuous.
Moreover, $f(\bbz)\leq f(\bbz^{(0)})$ implies that the functions $\|\bbY-\bbH\bbX\|_F^2$ and $\log(|S_{ij}|+\delta_1)$ are all bounded, rendering the domain of $f(\bbz)$ bounded. 
It follows then that the level set $\ccalZ^{(0)}$ is compact.

Finally, we need to prove that either \eqref{eq:step1_filterid} or \eqref{eq:step2_graph_denoising} has a unique solution, so that (\textit{\textbf{C4}}) is fulfilled. Prop. \ref{thm3} (see below) states that, under the two conditions required by Th. \ref{thm1} (i.e., $\bbS$ does not have repeated eigenvalues, and the graph signals $\bbX$ excite every graph frequency), the solution to \eqref{eq:step1_filterid} is unique. This confirms that (\textit{\textbf{C4}}) is satisfied, concluding the proof. 
%The conditions required for this to happen are stated in Th. \ref{thm3}.

\begin{proposition}\label{thm3}
Let $\bbH\in\reals^{N\times N}$, $\bbS=\bbV\diag(\bblambda)\bbV^{-1}\in\reals^{N\times N}$, and $\bbX\in\reals^{N\times M}$ be the GF, the GSO, and the input signals in \eqref{eq:step1_filterid}. Then, \eqref{eq:step1_filterid} has a unique solution w.r.t. $\bbH$ if the following conditions are satisfied:
\begin{enumerate}
	\item $\lambda_i \neq \lambda_{i'}$, for all $i\neq i'$ and $(i,i')\in\{1,...,N\}^2$.
	\item Every row of $\tbX = \bbV^{-1} \bbX$ has at least one non-zero entry. 
\end{enumerate}
\end{proposition}

\begin{proof}
To simplify exposition, we focus first on the (most restrictive) setup of having only $M=1$ input-output pairs. Defining $\hbh:=\vvec(\bbH)$, we can reformulate \eqref{eq:step1_filterid} as
\begin{equation}
	\text{min}_{\hbh \in \reals^{N^2}} \gamma\| (\bbI \otimes \bbS - \bbS^\top  \otimes \bbI) \hbh \|_2^2 + \| \bby - (\bbx^\top \otimes \bbI) \hbh \|_2^2,
	\label{eq:step1_filterid_vec}
\end{equation}
where lowercase symbols $\bby$ and $\bbx$ are used to emphasize that the output and input signals are a single $N$-dimensional vector. %
Upon defining $\bbD:= \bbI \otimes \bbS - \bbS^\top \otimes \bbI$, and $\bbE := \bbx^\top \otimes \bbI$, solving \eqref{eq:step1_filterid_vec} is equivalent to solving
\begin{equation}\label{eq:step1_filterid_single_ls_reformulation}
	\text{min}_{\hbh \in \reals^{N^2}} \Big\| \begin{bmatrix}\bbzero_{N^2} \\ \bby~\end{bmatrix} - \bbF \hbh \Big\|_2^2~\text{with}~\bbF:=\begin{bmatrix}\gamma \bbD \\ ~\bbE\end{bmatrix}
 \end{equation}
To prove that~\eqref{eq:step1_filterid_single_ls_reformulation} has a unique solution, it suffices to show that $\bbF$ is full column rank, i.e. $\nexists \; \bbn \in \reals^{N^2} $ such that $\bbF \bbn = \bbzero_{N+N^2}$. To show this, we first identify $\ccalN(\bbD)$, the null space of $\bbD$, and then show that $\bbE \bbn \neq \bbzero_N \; \forall \; \bbn \in \ccalN(\bbD)\setminus \{\bbzero_{N^2} \}$.

We start with the characterization of $\ccalN(\bbD)$. Given the Kronecker structure of $\bbD$, each of its $N^2$ eigenvalues has the form $\lambda_k - \lambda_{k'}$, with $(\bbV^{-1})^\top \otimes \bbV$ being the associated eigenvectors. Leveraging that $\lambda_i\neq \lambda_{i'}$ for $i\neq i'$, it follows that only when $i=i'$ the eigenvalue of $\bbD$ is zero. As a result, $\rank (\bbD) = N^2 - N$ and $\text{dim} (\ccalN(\bbD)) = N$. Equally important, the $N$ eigenvectors associated with the $N$ zero eigenvalues are given by $(\bbV^{-1})^\top \odot \bbV$, which, as a result, constitutes a basis spanning $\ccalN(\bbD)$. More formally, we concluded that $\ccalN(\bbD) = \{ ((\bbV^{-1})^\top \odot \bbV) \bbtheta \;| \forallsymb \; \bbtheta \in \reals^{N}\}$. 
%Nota that, by interpreting, it follows readily that represents all GFs. 
%As the commutativity term checks whether the eigenvectors of both $\bbS$ and $\bbH$ are the same, we can characterize $\ccalN(\bbD) = \{ ((\bbV^{-1})^\top \odot \bbV) \bbtheta \; \forallsymb \; \bbtheta \in \reals^{N},\; \bbtheta \neq \bbzero_N \}$, where we have the basis represented by the $N^2 \times N$ matrix obtained from the eigenvectors of $\bbS$.
%Note that this represents the definition of a GF in the frequency domain, thus sharing eigenvectors with $\bbS$ and rendering 0 the commutativity term in \eqref{eq:step1_filterid}.
%\green{[VMTG]: tengo una demostración de que $\bbD((\bbV^{-1})^\top \odot \bbV) = \bbzero \in \reals^{N^2 \times N}$, pero es mucho más larga y no sé hasta qué punto merece la pena. Si queréis que la escriba decidme.}

%Now that we have fully characterized $\ccalN(\bbD)$ when $\bbS$ does not have repeated eigenvalues.
Thus, to show that $\bbF$ in~\eqref{eq:step1_filterid_single_ls_reformulation} is full column rank we just need to prove that the only element $\bbn\in\ccalN(\bbD)$ that renders $\bbE \bbn = \bbzero_N$ is the all-zero vector $\bbzero_{N^2}$. To do so, we leverage the characterization of $\ccalN(\bbD)$ and write $\bbE \bbn$ as
\begin{align}
	\nonumber \bbE \bbn &= (\bbx^\top \otimes \bbI) ((\bbV^{-1})^\top \odot \bbV) \bbtheta % \neq \bbzero \in \reals^{N^2} \\
	= (\bbx^\top (\bbV^{-1})^\top \odot \bbV) \bbtheta \\% \neq \bbzero,
	\nonumber &= \bbV \text{diag}(\bbtheta) (\bbx^\top (\bbV^{-1})^\top)^\top %\neq \bbzero \in \reals^{N \times N} \\
	\nonumber = \bbV \text{diag}(\bbtheta)  \bbV^{-1} \bbx \\%\neq \bbzero \\
	&= \bbV \text{diag}(\bbtheta) \tbx%\neq \bbzero \\
	= \bbV (\bbtheta \circ \tbx), \label{eq:Dtimesn_cannotbezero}
\end{align}
where we used the property $(a \otimes b) (c \odot d) = a c \odot b d$.
Since $\bbV$ is invertible, the first and last terms in \eqref{eq:Dtimesn_cannotbezero} demonstrate that $\bbE \bbn=\bbzero_N$ requires $\bbtheta \circ \tbx=\bbzero_N$. However, condition 2) in Prop.~\ref{thm3} states that $\tilde{x}_i \neq 0 ~\forall~ i$; hence, $\bbtheta \circ \tbx=\bbzero_N$ requires $\bbtheta=\bbzero_N$.  This implies that the only element in $\ccalN(\bbD)$ that renders $\bbE \bbn = \bbzero_N$ is $\bbn=(\bbV^{-1})^\top \odot \bbV) \bbzero_N =\bbzero_{N^2}$, concluding the proof.

The proof can be generalized for $M > 1$. In that case, the matrix $\bbE$ has size $MN\times N^2$ and the counterpart to \eqref{eq:Dtimesn_cannotbezero} establishes that having $\bbE \bbn=\bbzero$ requires $\vvec(\text{diag}(\bbtheta) \tbX)= \bbzero_{MN}$. Since Prop.~\ref{thm3} assumes that each row of $\tbX$ has at least one nonzero entry, it follows that $\bbtheta=\bbzero_{MN}$, concluding the proof. 
%only the steps associated xxxx follows analogous steps, it is enough if every frequency $k$ is excited by at least one graph signal, as this ensures that the product $\text{diag}(\bbtheta) \tbX \neq \bbzero_{N \times M}$ for every $\bbtheta \neq \bbzero_N$, concluding the proof.
\end{proof}

\bibliographystyle{IEEEtran.bst}

\bibliography{bibliography}

\end{document}